\documentclass[runningheads]{llncs}
\usepackage[T1]{fontenc}
\usepackage{amsmath}
\usepackage{hyperref}
\usepackage{amssymb}
\usepackage{ifsym}
\usepackage{stmaryrd}
\usepackage{listings}
\usepackage{multicol}
\usepackage{multirow}
\usepackage{boldline}
\usepackage{array}
\usepackage{graphicx}
\usepackage{tikz}
\usepackage{appendix}
\usepackage{hyperref}
\usepackage{amssymb}
\usepackage{ifsym}
\usepackage{stmaryrd}
\DeclareMathOperator{\cop}{copy}
\DeclareMathOperator{\opt}{opt}
\DeclareMathOperator{\leftjoin}{leftjoin}
\DeclareMathOperator{\includes}{incl}
\DeclareMathOperator{\notIncludes}{notIncl}
\DeclareMathOperator{\conj}{conj}
\DeclareMathOperator{\eval}{eval}

\DeclareMathOperator{\equalTo}{equalTo}
\DeclareMathOperator{\notEqualTo}{notEqualTo}

\DeclareMathOperator{\result}{result}
\DeclareMathOperator{\lis}{tl}
\DeclareMathOperator{\unio}{union}
\DeclareMathOperator{\unnest}{unest}

\DeclareMathOperator{\impl}{impl}
\DeclareMathOperator{\type}{type}
\DeclareMathOperator{\Person}{Person}
\DeclareMathOperator{\Researcher}{Researcher}

\usepackage{listings}
\usepackage{fancyvrb}
\newcommand\listingsize{\fontsize{9pt}{10pt}}
\RecustomVerbatimCommand{\Verb}{Verb}{fontsize=\listingsize}
\RecustomVerbatimEnvironment{Verbatim}{Verbatim}{fontsize=\listingsize}
\definecolor{grey}{RGB}{130,130,130}
\definecolor{darkgrey}{RGB}{80,80,80}

\lstdefinelanguage{SPARQL}{
    keywords={SELECT, ASK, CONSTRUCT, DESCRIBE, WHERE, UNION, OPTIONAL, FILTER, DISTINCT, LIMIT, OFFSET, PREFIX, BASE, ORDER, BY, ASC, DESC, GROUP, HAVING, COUNT, SUM, MIN, MAX, AVG, SAMPLE, BIND, VALUES, MINUS, EXISTS, NOT, IN, STR, LANG, LANGMATCHES, DATATYPE, BOUND, COALESCE, IF, SAMETERM, ISIRI, ISURI, ISBLANK, ISLITERAL, NOW, YEAR, MONTH, DAY, HOURS, MINUTES, SECONDS, TIMEZONE, TZ, BNODE, RAND, ABS, CEIL, FLOOR, ROUND, CONCAT, STRLEN, UCASE, LCASE, ENCODE_FOR_URI, CONTAINS, STRSTARTS, STRENDS, STRBEFORE, STRAFTER, REPLACE, SUBSTR, REGEX, TRUE, FALSE, SERVICE, GRAPH, NAMED},
    keywordstyle=\color{blue}\bfseries,
    sensitive=true,
    comment=[l]{\#}, % SPARQL comments start with #
    morestring=[b]", % Strings in double quotes
    morestring=[b]', % Strings in single quotes
    stringstyle=\color{red}\ttfamily, % String style
    identifierstyle=\color{black}, % Identifiers
    basicstyle=\ttfamily\small, % Base font style
    morecomment=[s]{/*}{*/}, % Multi-line comments
    emph={a}, emphstyle=\color{blue}, % Highlight 'a' as a shorthand keyword for rdf:type
    literate={?}{$\mathsf{?}$}{1} {*}{{\color{orange}$\mathsf{*}$}}{1}, % Variables and wildcards
}

\newcolumntype{?}{!{\vrule width 1pt}}

\begin{document}
\title{SPARQL in N3: SPARQL \textsc{construct} as a rule language for the Semantic Web \\(Extended Version)}
\titlerunning{SPARQL in N3}
\author{Dörthe Arndt\inst{1,2} \and
 William Van Woensel\inst{3} \and
 Dominik Tomaszuk\inst{4}}
% %
 \authorrunning{D. Arndt et al.}
 \institute{Computational Logic Group, Technische Universität Dresden, 
 Dresden, Germany 
% \url{http://www.springer.com/gp/computer-science/lncs} 
\and
 ScaDS.AI, Dresden/Leipzig, Germany\\
 \and
 Telfer School of Management, University of Ottawa, Ottawa, Canada\\
 \and
 University of Bialystok, Bialystok, Poland
 }
\maketitle              % typeset the header of the contribution

\begin{abstract}

Reasoning in the Semantic Web (SW) commonly uses Description Logics (DL) via OWL2 DL ontologies, or SWRL for variables and Horn clauses. The Rule Interchange Format (RIF) offers more expressive rules but is defined outside RDF and rarely adopted. For querying, SPARQL is a well-established standard operating directly on RDF triples. We leverage SPARQL \textsc{construct} queries as logic rules, enabling (1) an expressive, familiar SW rule language, and (2) general recursion, where queries can act on the results of others. We translate these queries to the Notation3 Logic (N3) rule language, allowing use of existing reasoning machinery with forward and backward chaining. Targeting a one-to-one query–rule mapping improves exchangeability and interpretability. Benchmarks indicate competitive performance, aiming to advance the potential of rule-based reasoning in the SW.

\keywords{SPARQL  \and Notation3 \and Rule-based reasoning.}
\end{abstract}

%~\input{intro}
%\input{intro2}
\section{Introduction}
\label{sec:intro}
In our experience, most reasoning on the Semantic Web (SW) relies on Description Logics (DL) in the form of OWL2 DL \cite{motik2012} ontologies. DL is less expressive than first-order logic, and expressing the desired (e.g., rule-based) logic can be challenging. 
SWRL~\cite{horrocks2004} extends OWL interpretations with variables and Horn-like rules,  %and was meant to allow rule-based logic on top of OWL. 
but operates on OWL rather than RDF, and cannot generate blank nodes.
The more expressive Rule Interchange Format (RIF) offers 3 dialects to cover the different needs of rule systems. However, it is defined outside of the RDF data model, and is, to the best of our knowledge, not widely adopted. 
We thus argue that the potential of logic SW reasoning has not been realized.

On the other hand, when it comes to querying RDF data, SPARQL~\cite{harris2013sparql} is an expressive and well-established standard, with native support for RDF triples. 
Notably, since SPARQL \textsc{construct} queries are closed, i.e., RDF graphs are used as both input and output, they can be used to represent logical rules: if the \textsc{where} clause holds, then the query template can be inferred.
The %Hence, we can leverage the 
expressivity and popularity of SPARQL can be leveraged to express logic rules %, which many users will already understand, 
on the SW. 
Moreover, this adds general recursion in SPARQL: i.e., when executed as rules, \textsc{construct} queries can operate on the results of other \textsc{construct} queries. In this work, we support SPARQL~1.1 query patterns, filters, and the \textsc{select} and \textsc{construct} forms. 
%Property paths are treated as syntactic sugar, since their semantics can be captured by recursively applying our rules.}

Others have proposed using SPARQL as a logic rule language. Two types of approaches exist: (a) translating SPARQL into a non-SW logic rule language, namely Datalog~\cite{polleres2007,gottlob2015beyond,angles2023}, so that rule engines can be used for reasoning; 
and (b) extending an existing query engine with reasoning capabilities~\cite{reutter2021,knublauchSPIN}.
However, Datalog does not follow SW principles; 
%, meaning there is an impedance mismatch with RDF, and translations are not easily exchangeable.
and while query engines allow dealing with large datasets and many joins, they lack backward reasoning capabilities.
% Its \textsc{construct} query form enables rule-like constructs in a syntax familiar to SW users, facilitating modular graph pattern definitions.
% Despite its expressivity, SPARQL has limitations in enabling rule-based reasoning. Recursive constructs, essential for advanced reasoning, are not natively supported beyond property paths~\cite{kostylev2015,reutter2021}. Existing approaches address these gaps by translating SPARQL queries into rule-based languages like Datalog. However, these translations fail to align with SW principles, introducing compatibility issues and impeding the exchange of rules across systems. Specifically, Datalog does not natively handle RDF triples or employ IRIs for resource identification, creating a semantic disconnect between SPARQL and the target reasoning environment.

We propose a translation of SPARQL into the Notation3 (N3)~\cite{N3Logic,woensel2023n3} logic rule language, called SiN3 (SPARQL-in-N3).
This allows leveraging the existing rule machinery (eye~\cite{eyepaper}, cwm~\cite{cwm}, jen3~\cite{jen3}) for implementing backward and forward reasoning reasoning.
Moreover, as N3 follows SW principles, translations are not restricted to local use and can be exchanged on the SW. 
% there is native support for triples as atomic formulas, and translations are more easily exchangeable. 
% To further improve exchangeability, 
% Moreover, N3 has a powerful set of SW-compliant builtins that can be mapped to SPARQL builtins.
% In this vein, we also translate a single query into a single rule, which makes it easier to exchange, and even manually tweak, translations. 
% (wvw) prior sentence sounded awkward
A ``runtime'' set of rules, 
% We execute translated rules using the concept of a ``runtime'', 
which is novel in the literature on SPARQL to rule translation, allows (1) translating a single query into a single rule, which makes it easier to exchange, and even manually tweak, translations; and even (2)
% enables the execution of these 1-1 translated rules.
a true cross-fertilization, by also allowing SPARQL features to be used in N3.
To illustrate its meta-reasoning capabilities, we used N3 itself to translate %implement the translation process 
from a triple-based SPARQL (SPIN~\cite{knublauchSPIN}) to N3 rules.
Our evaluation %, using datasets from the literature, 
shows a competitive performance for our preliminary work. % and contrasts the performance of query vs. rule engines. 
% A custom use case further illustrates modular queries, as afforded by general recursion in SPARQL.
A demo paper was previously published about this work~\cite{sin3_demo}.

% To address these challenges, we propose a novel integration of SPARQL with Notation3 Logic (N3) \cite{N3Logic,woensel2023n3}, a rule language that extends RDF and adheres to SW principles. N3 enables seamless reasoning with RDF triples and offers native support for rule-based constructs. By translating SPARQL queries into N3 rules, we preserve SPARQL’s expressivity, including \textsc{construct} queries, property paths, and \textsc{optional} constructs, while enabling advanced rule-based reasoning. Our approach ensures that rules remain compatible, exchangeable, and executable within SW frameworks.

% This paper presents a translation framework from SPARQL to N3, targeting a one-to-one mapping between queries and rules. This framework supports runtime reasoning by providing a “runtime” ruleset for real-time introspection of recursive query execution. Benchmarks from state-of-the-art datasets demonstrate competitive performance, showing that our approach can bridge the gap between SPARQL querying and rule-based reasoning.

\textit{Contributions}. (1) We translate SPARQL into the SW-compliant N3 rule language; (2) In doing so, we add general recursion to SPARQL, a feature that is still lacking~\cite{reutter2021,angles2023}; (3) We provide correctness results for our approach; (4) We evaluate and compare our approach to other similar systems, using large datasets from the literature. % and a custom dataset.
As far as we know, this is the first effort to translate SPARQL into a SW rule language, with a one-to-one mapping between queries and rules. 

% \paragraph{Contributions.} 
% In this paper, we provide a seamless translation of SPARQL queries into Notation3 Logic (N3), enabling rule-based reasoning while adhering to Semantic Web principles. Our approach introduces runtime introspection for recursive queries, allowing users to monitor recursion depth and execution progress. By integrating SPARQL features such as property paths, \textsc{optional}, and aggregation into a rule-compatible RDF reasoning environment, we ensure compatibility and flexibility. Additionally, we demonstrate the viability of our solution through an open-source implementation and performance benchmarks, highlighting its competitiveness for advanced Semantic Web reasoning tasks.

\textit{Paper Organization.}
%The remainder of this paper is organized as follows. 
Section~\ref{sec:bground} introduces the overall idea of our approach. Section~\ref{sec:rel_work} discusses related work. %We the introSections \ref{sec:sparql} and \ref{sec:n3} introduce SPARQL and N3. 
Sections~\ref{sec:sparql} and \ref{sec:n3} explain SPARQL, and N3, respectively. 
Section \ref{sec:sparql2n3} describes our translation of SPARQL queries into N3 rules, including complex constructs and runtime rules.
%Section~\ref{sec:imp} details the implementation. 
Section~\ref{sec:eval} evaluates the performance of our method. 
Finally, Section~\ref{sec:conclusions} concludes the paper. % and outlines future directions.

\section{Motivation}
\label{sec:bground}
%\wvw{Will need to make sure there's not too big of a leap from this section to semantic sections (e.g., already mention variable mappings as solutions?)}
%\dt{I still think this section is too obvious in the context of ESWC. Maybe consider a fragment on N3, but it will be short enough to be moved to another section. Especially since we need the space!}
%\da{I will rewrite in a way that we put our running example here and also have the overall idea detailed in this section. Then, we go back to it later on and do not always need to come up with new examples. I agree that we do not need to introduce RDF}

To motivate our choice for logic rule language, 
% In order to explain the main idea of our contribution, 
we provide examples of SPARQL and N3 which illustrate how closely these two frameworks are related. 
% This relatedness motivates the idea of mapping from one to the other and to thereby close the gap between these two. %We assume the reader to be familiar with RDF~\cite{rdf11-concepts}. Our 
%A more formal treatment of these frameworks and our mapping is provided in subsequent sections (Sections \ref{sec:sparql},\ref{sec:n3} and \ref{sec:sparql-n3}).

%In RDF, a resource is anything that can be described, albeit an abstract concept or physical object. To identify resources in RDF, IRI (Internationalized Resource Identifier) and literal terms (e.g., numbers and strings) can be used.
%Resources are described using triples, also known as statements, where the \textit{subject} represents the resource being described; the \textit{predicate}, also known as the property, defines the relationship or attribute of the subject; and the \textit{object} represents the value or another resource related to the subject.
%For example, consider the RDF triple in Listing~\ref{lis:example}, which defines an individual's age.
SPARQL and N3 both operate on RDF triples like for example:\footnote{For brevity, we omit prefixes if they are not relevant or well-known.
%provide a table of all prefixes used throughout the paper in Appendix~\ref{sec:prefixes}.
}
\begin{equation}\label{john}
\texttt{:John a :Researcher.}
\end{equation}
%Informally, this triple means that \textit{"John is a researcher."}. 
With SPARQL we can query for patterns on sets of triples and retrieve variable bindings (\textsc{select}) or create new graphs (\textsc{construct}).
% (wv) already mentioning this in evaluation
%\footnote{In this work we do not cover \textsc{ask} and \textsc{describe} queries. In appendix we provide a full list of all feature we currently support. \todo{add} } 
Applying the query 
\begin{equation}\label{construct}
	\texttt{CONSTRUCT \{?x a :Person . \} WHERE \{?x a :Researcher . \} }
\end{equation}
on a graph containing triple \ref{john} results in:
\begin{equation}\label{john2}
\texttt{:John a :Person.}
\end{equation}
Since a new graph is being constructed from the original graph, this is comparable to
%step could be seen as reasoning; more concretely, as 
the application of the following N3 rule  
% Given that a new graph is constructed based on the information present in the original graph, this step of producing a new graph could be seen as reasoning, more concretely, as the application of a rule like  
\begin{equation}\label{n3rule}
	\texttt{\{?x a :Researcher .\}=>\{?x a :Person .\}}
\end{equation}
where \texttt{=>} is a logical implication; the conclusion (second clause) is inferred if the premise holds (first clause). 
% This last expression is exactly how we write rules in N3 and 
Applying rule \ref{n3rule} on triple \ref{john} also yields triple \ref{john2}. 

An important difference between executing rule~\ref{n3rule} and query~\ref{construct}, however, is that N3 applies rules recursively: 
%during the reasoning process with rule \ref{n3rule}, 
triple \ref{john2} becomes part of the knowledge, and other rules can be applied on this derivation. 
In contrast to that, SPARQL is a query language and only produces triple \ref{john2} to provide it as output to the user. 

When choosing a logic rule language, 
% When aiming for an implementation of recursive SPARQL \textsc{construct} queries, 
the close relation in terms of syntax between SPARQL and N3 makes the latter a good candidate: users coming from SPARQL, while they will likely not be N3 experts, will be better able to understand and---if needed---tweak the N3 translation. 
% To ensure that this is really the case, 
To further facilitate such manual editing, as well as the exchange of translations on the Web,
we translate each query to exactly one N3 rule. 
Users can also
%easily exchange the  translated rules through the Web, and 
combine these rules with existing, ``native'' N3 rules from other sources to suit their needs (e.g., OWL2 RL \cite{fokoue2012owl}, OWL-P \cite{tomaszuk2018inference} and L2 \cite{fischer2010towards} inferencing). % and, -- by using  reasoners (e.g., EYE~\cite{eyepaper} or cwm \cite{cwm}) -- get explanations for the conclusions drawn from their queries.  

%As SPARQL is so well established in the Web, we think that 
%If we now want to offer recursive \textsc{construct} queries simply because many 
%Our goal is to provide a semantic Web rule language which is easy to learn for the user 

%The motivation for this paper is now, that we would like enable users who are not familiar with writing rules but who have experience with SPARQL to write and apply rules.
%We would futhermore like them to be able to understand the derivations done 
%\da{idea: maybe also bring proofs in, but I am in the middle of writing and will think %about it while biking home ;)}
%We can easily see that N3 rules and SPARQL construct queries are closely related. The main difference between the \textsc{construct} query and an N3 rule is, however, that rules are   

\section{Related Work}
\label{sec:rel_work}
% \da{we switched focus, but I think our related work need to get more to the point, e.g., how the rule languages are not native to the web, but not sure here}
% Recursion is a fundamental feature in many query languages, including SQL (first appearing in SQL:1999.)~\cite{kostylev2015}. 
% While recursive functionality in SPARQL is not a new concept, implementations have varied widely in their approach. 
% SPARQL 1.1 introduced property paths as a recursive operator to facilitate RDF graph navigation, but their expressivity is limited~\cite{kostylev2015,reutter2021}. 
% Therefore, the need for extended recursion capabilities in SPARQL persists~\cite{reutter2021}. 
Below, we review approaches that utilize SPARQL to perform reasoning.
% by translating SPARQL into logic rule languages (Subsection~\ref{sec:sparql2rules}), and by extending SPARQL engines with reasoning capabilities (Subsection~\ref{sec:recurs_engines}).

\subsection{SPARQL Translation to Rule Languages}
\label{sec:sparql2rules}

%Several studies translate SPARQL into logic rule languages, gaining recursive capabilities due to the translated language. 
Polleres~\cite{polleres2007} first translated SPARQL into Datalog, with negation as failure based on Answer Set Programming (ASP). 
It allows blank node-free \textsc{construct} queries to act as a recursive rule language over RDF.
% Polleres highlights the challenges in translating SPARQL's \textsc{optional} clauses to Datalog and introduces three semantics that progressively define the compatibility of variable mappings. 
% The approach does not support certain SPARQL constructs, such as \textsc{construct} queries with blank nodes, solution modifiers, and complex \textsc{filter}s, which remain incomplete.
Gottlob et al.~\cite{gottlob2015beyond} target general recursion in SPARQL by
%address the limitations in SPARQL 1.1 property paths by advocating for general recursion and navigation capabilities within query languages. They 
% introduce TriQ-Lite, based on $Datalog^{\exists,\neg s,\bot}$, to offer a balance between expressivity and tractable evaluation. To simplify syntax and semantics, they further 
developing TriQ-Lite 1.0 based on Warded $Datalog^{\exists,\neg sg,\bot}$.
The authors translate SPARQL queries under the OWL2 QL entailment regime into TriQ-Lite 1.0.
% This variant addresses existential quantification in rule heads, supporting a translation of SPARQL queries under the OWL2 QL entailment regime, although no implementation exists.
Angles et al.~\cite{arenas2018expressive} present SparqLog, a SPARQL translation engine built on Vadalog~\cite{bellomarini2018vadalog}. 
% SparqLog is meant to cover requirements from the database and SW community, including coverage of SPARQL features, multiset semantics, ontology reasoning, recursion, and more.
Similar to Gottlob et al., they rely on Warded Datalog, and both support OWL2 QL entailment. 
However, recursive querying is not directly supported; queries are translated using overlapping predicate names, and without referring to the results of other queries.

In general, we point out that Datalog does not follow SW principles, such as the use of IRIs for resource identification, and using triples as atomic formulas. Hence, translations are not easily exchangeable, and there is an impedance mismatch as RDF triples first need to be translated into predicates.
% Utilizing Warded $Datalog\pm$, it supports OWL2 QL entailment and facilitates RDF dataset querying without needing a separate translation step. Despite its experimental success with SPARQL benchmarks, SparqLog lacks full support for constructs like \textsc{construct}, \textsc{bind}, and \textsc{values}.
% \da{Do we need to be careful with the phrasing here and only mention Construct? we also don't support the other features}

\subsection{Extensions to SPARQL Engines}
\label{sec:recurs_engines}

% Various approaches extend SPARQL engines with recursion by introducing new syntax. These efforts allow SPARQL to reference the results of recursive queries, facilitating queries on datasets requiring complex iterative reasoning.

Reutter et al.~\cite{reutter2021} propose recSPARQL, which introduces a fixed-point operator, similar to the one from SQL:1999, into SPARQL. 
Using \textsc{construct} queries, a temporal graph is iteratively populated until reaching the least fixed point (forward reasoning). 
% This CONSTRUCT query, or a subsequent SELECT or ASK query, can refer to the temporal graph to act on prior query results.
The authors focus on \textsc{construct} queries without blank nodes in the template part, and require them to have fixed points.
A restriction to linear recursion~\cite{reutter2021} means an iteration can only rely on the original data and triples from the prior iteration. The system was implemented on top of Apache Jena~\cite{Apache2021} (ARQ module). %; the implementation is available online\footnote{https://adriansoto.cl/RecSPARQL/}. 
% The approach implemented atop Apache Jena’s ARQ module, enforces monotonicity by disallowing explicit negation and recursive references in \textsc{optional} clauses. 
% This design choice ensures fixed-point consistency while supporting RDF data through constructs aligned with $Datalog^{rbr}_{rdf}$.
% (wv) just shortening this a bit
SPIN~\cite{knublauchSPIN} (SPARQL Inferencing Notation) allows representing rules and constraints in RDF using SPARQL.
The spinrdf~\cite{spinrdf} library supports reasoning with SPARQL \textsc{construct} queries, It also uses a fixed point (forward reasoning) algorithm and 
%The library 
is also implemented on top of Apache Jena.
% extends SPARQL with rules and constraints while leveraging existing RDF technologies. By embedding executable SPARQL queries, it allows for representing business rules and data validation directly within RDF graphs. SPIN supports various use cases, including calculating derived property values, enforcing data constraints, and incrementally reasoning over RDF data.

\sloppy
Hogan et al.~\cite{hogan2020recursive} introduce SPARQAL, which combines graph querying with analytical operations, including \textsc{do...while} loops.  
% These loops iterate until a termination condition is met, enhancing SPARQL’s analytical capacity with recursive Turing-complete querying. Implemented on Apache Jena, SPARQAL performs efficiently by using batch joins akin to Map-Reduce, which are suitable for data-intensive recursive queries.
Corby et al.~\cite{corby2017ldscript} extend SPARQL with LDScript, an imperative language that introduces functions and loops. 
While these loops may allow for recursion on query results, it requires an imperative extension that we consider outside the scope of this paper. 
Atzori et al.~\cite{atzori2014} propose a SPARQL function for resolving recursive expressions.
% This flexibility allows recursive constructs directly within SPARQL queries, facilitating complex reasoning while maintaining alignment with SPARQL’s filter-based query processing.
% Finally, we note that several other works focus on extending the graph navigation capabilities of SPARQL~\cite{perez2010nsparql,kochut2007sparqler,libkin2013trial}. 
% Perez et al. proposes nSPARQL for graph navigation~\cite{perez2010nsparql}, with a syntax similar to XPath.
% Kochut et al. introduce SPARQLeR (SPARQL Extended with Regular paths)~\cite{kochut2007sparqler}, where triple patterns can have \%path variables as predicates; a path variable matches any path between two nodes. 
% Other notable works include Atzori et al.~\cite{atzori2014}, who propose a function-based recursive method within SPARQL queries, and Perez et al.~\cite{perez2010nsparql} with nSPARQL for graph navigation. 
% TriAL* by Libkin et al.~\cite{libkin2013trial} moves from graph to RDF navigation, i.e., where predicates may also act as subjects/objects of other triples.
% further explores recursion for triple joins, adding depth to recursive self-joins for extended pathfinding across RDF graphs. These contributions highlight an evolving landscape where recursion in SPARQL continues to expand through both translation to rule-based systems and internal extensions within SPARQL engines.
\section{SPARQL}
\label{sec:sparql}
In this section, we give a short introduction to the main features of SPARQL. 
Our definitions follow the work of Pérez et al.~\cite{perez2009semantics} and extend them to accommodate the features introduced in SPARQL 1.1, drawing on related works~\cite{polleres2013relation,salas2022semantics}. A more detailed introduction can be found there.

\subsection{Syntax}
\label{sec:sp_syn}
Atomic terms in SPARQL are taken from the sets $I$, $B$, $L$, and $V$, representing IRIs, Blank Nodes, RDF literals, and Variables, respectively. A filter condition $R$ is a (logical) formula which depends on a (possibly empty) set of atomic terms.
%which typically also involves functions such as \texttt{isBlank} or \texttt{isIRI}.
The set of filter conditions contains the formulas $\top$ (true) and $\bot$ (false). %\footnote{This is in particularly relevant for the otional operator OPT. When this comes with only two arguments in the concrete syntax, we use $\top$ as its filter condition. }

We recursively define the set $GP$ of graph patterns as follows:
%Graph Pattern expressions (denoted $GP$) are recursively defined as follows:

%\todo{question: BGP with or without blank nodes. I think it should be with, but then we need to add a mapping to the semantics.}
%\dt{If clarity and explicit data references are paramount, opt for BGPs without blank nodes. If you only need to ensure the existence of unnamed resources and can handle potential duplicates, blank nodes can simplify the query.}
\vspace{0.15cm}
\noindent 
(1) If  $P \subset (I  \cup L \cup V) \times (I \cup V) \times (I  \cup L \cup V)$, then $P\in GP$.  We call $P$ a Basic Graph Pattern (BGP)\footnote{In practice, BGPs also allow blank nodes in subject and object position. We omit them here to keep the semantics concise. Note that our translation in Section \ref{sec:sparql-n3} also works if we allow blank nodes in BGPs as these are locally scoped in  N3.},
%is a graph pattern (a.k.a. triple pattern).
% If the tuple does not known as a triple pattern, where terms may include variables. 
%A set of these forms a Basic Graph Pattern (BGP).
\\
%    \item[(2)] A tuple $(s, PP, o)$ is a property path pattern, where $PP$ represents a property path expression. This pattern allows for flexible querying over property paths in RDF graphs.
%\sloppy
(2) If $P, P_1, P_2 \in GP$ and $R$ is a filter condition, then: 
\begin{itemize}
\item	$\text{AND} (P_1, P_2), \text{UNION} (P_1,  P_2), \text{MINUS}(P_1, P_2)\in GP$ 
\item $\text{FE}(P_1, P_2), \text{FNE}(P_1, P_2)\in GP$,
\item $\text{OPT}(P_1, P_2, R) \in GP$,
\item  $\text{FILTER}(P, R)\in GP$.
\end{itemize}
%$\text{AND} (P_1, P_2), \text{UNION} (P_1,  P_2), \text{MINUS}(P_1, P_2),
% are graph patterns. \\ 
 %representing conjunction, optional, union, minus, filter exists and filter not exists, respectively.\\
% (3) If $P_1, P_2 \in GP$ and $R$  a filter condition, 
%then
%\text{OPT}(P_1, P_2, R), \text{FE}(P_1, P_2, R), \text{FNE}(P_1, P_2, R) \in GP$, % are graph patterns.

%(4) If $P\in GP$ and $R$ is a SPARQL condition, then
 %$\text{FILTER}(P, R)\in GP$. % is a graph pattern.
%    \item[(4)] If $P$ and $Q$ are graph patterns  then $(P \text{ FILTER } R)$ applies the filter condition $R$ to $P$.
%    \item[(5)] $(P \text{ BIND } exp \text{ AS } x)$ is a graph pattern that assigns the result of an expression $exp$ to the variable $x$.
\vspace{0.15cm}
%As we will further elaborate below, the evaluation of the condition $R$ rely on variable binding and can be true, false or error depending 
Note that we define
%The SPARQL condition $R$ is a formula which depends on variables and can be evaluated to true, false or error depending on the bindings of these variables. 
%The SPARQL condition can be  understood as a function on variebles which based on the 
%This view on the FILTER condition also explains why we list 
 FILTER EXISTS (FE) and FILTER NOT EXISTS (FNE) as separate graph patterns even though they are special cases of the FILTER condition. This is because their evaluation depends on both the variable bindings and the graph the query is evaluated against, while all other filter conditions only depend on the former~\cite[Section 17.4.1.4]{harris2013sparql}. 
We still assume that the translation form concrete SPARQL syntax into the algebra \cite[Section 18.2]{harris2013sparql} treats FE and FNE as filters when producing the  algebra expression OPT for the OPTIONAL operator. Here, we also remind the reader that if there is no FILTER condition occuring on the second argument of an OPTIONAL expression in conrete syntax, $\top$
is added as filter argument during the translation.
 %We furthermore defined the optional operator OPT with three arguments. Whenever only two are used
 % is that their evaluation does not only rely 
 %, the evaluation of FE and FNE does not only rely 
% on variable bindings but also on the graph the query is evaluated against. %As  we will further elaborate below, that is because their evaluation depends on the  dataset the queery is evaluated against. This is different for other filter conditions which we here represent by the symbol
%
%$R$. The condition $R$ only relies on variables and can be evaluated to ture, fals or error depending on the bindings of these variables.  %equality symbols, and unary predicates such as \texttt{bound}, \texttt{isBlank}, and \texttt{isIRI}. 
%We do not detail the different conditions which could occur in a FILTER here and refer the interested reader to Pérez et al.~\cite{perez2009semantics}.

%For any pattern $P$, we denote by $var(P)$ the set of all variables occurring in 
%within P, that is variables only occurring in filter expressions do not count among vars(P). Analogously, for any filter expression R, we denote by vars(R) the set of all variables occurring in
%R.

\subsection{Semantics}
\label{sec:sparqlsemantics}
%\todo{decide: set or multiset?}
%\dt{I can go with adopting bag semantics during query evaluation for compliance with SPARQL 1.1 and to support the correct behavior of aggregates and solution modifiers. However, the output of CONSTRUCT queries should conform to set semantics, ensuring that the resulting RDF graph contains unique triples, as required by the RDF data model.}
Let $T =I \cup B \cup L$ be the set of nonvariable terms. 
The semantics of SPARQL graph patterns rely on partial mappings from $V$ to $T$. 
%We get the application $\mu(t)$ of  a mapping $\mu : V \rightarrow T$  on a triple $t\in BGP$ by replacing all variables $v$ occuring in $t$ by their value $\mu(v)$. 
The domain of a mapping $\mu$, denoted $\text{dom}(\mu)$, is the set of variables where $\mu$ is defined. Two mappings, $\mu_1$ and $\mu_2$, are compatible, written as as $\mu_1\sim \mu_2$, if, for any $x \in \text{dom}(\mu_1) \cap \text{dom}(\mu_2)$, $\mu_1(x) = \mu_2(x)$.  Below, we consider mappings as sets of pairs and therefore use the  union operator $\cup$ to combine them.

Given sets of mappings $\Omega_1$ and $\Omega_2$, we define:
\begin{itemize}
    \item $\Omega_1 \bowtie \Omega_2 = \{\mu_1 \cup \mu_2 \mid \mu_1 \in \Omega_1, \mu_2 \in \Omega_2 \text{ and } \mu_1\sim\mu_2\}$, %\todo{DA: mention that mappings are considered sets of pairs.}
    \item $\Omega_1 \cup \Omega_2 = \{\mu \mid \mu \in \Omega_1 \text{ or } \mu \in \Omega_2\}$,
    \item $\Omega_1 \setminus_M \Omega_2 = \{\mu \in \Omega_1 \mid \text{for all } \mu^\prime \in \Omega_2: \mu \not\sim\mu^\prime\text{ or } \text{dom}(\mu)\cap \text{dom}(\mu^\prime)=\emptyset  \}$.
\end{itemize}

We apply a mapping $\mu$ to a BGP or filter condition $Q$, denoted by $Q\mu$,
%\wvw{I assume the Q here refers to either a BPG or filter condition?} \da{yes, but we literally say that just before}\wvw{It seemed the Q only applied to the filter, hence my question; I removed the ``a'' before filter that was confusing me, but I wonder if others could still misread it.} 
by replacing all variables $v\in \ dom(\mu)$ occurring in $Q$ by $\mu(v)$.  If $R\mu=\top$ for a filter condition $R$ and a mapping $\mu$, we also  write $\mu\models R$ 
%\wvwdone \wvw{Why use a different symbol R here?}.\da{because this time it is only for filters}\wvw{Ah and you use R for filters below - got it.}
%The left outer join is defined as: $\Omega_1 \tiny{\textifsym{d|><|}} \; \Omega_2 = (\Omega_1 \bowtie \Omega_2) \cup (\Omega_1 \setminus \Omega_2)$.
%The filter is a logical formula $R$ containing free variables. We say that $R$ is true under a mapping $\mu$, written as $\mu\models R$, if the formula in which all variables are replaced by terms according to $\mu$ evaluates to true.

\vspace{0.2cm}
\noindent \textbf{Evaluation of SPARQL Graph Patterns}
%\footnote{As this paper mainly focuses on SPARQL CONSTRUCT, we follow bag semantics here.}
The evaluation of graph patterns over an RDF graph $G$ produces a set of mappings. The function $\llbracket P \rrbracket_G$ denotes the evaluation of a graph pattern $P$ on $G$, and is defined as follows:
\begin{enumerate}
	\item[(1)] If $P$ is a BGP, then $\llbracket P \rrbracket_G = \{\mu \ | \ dom(\mu) = var(P) \text{ and } P\mu \subseteq G\}$, %where $dom(\mu)$ is the domain of $\mu$,
    %and $var(t)$ is the set of variables occurring in $t$, 
	\item[(2)] If $P$ is $\text{AND}(P_1, P_2)$, then $\llbracket P \rrbracket_G = \llbracket P_1 \rrbracket_G \bowtie \llbracket P_2 \rrbracket_G$,
	\item[(3)]
 If $P$ is $\text{OPT}(P_1, P_2, R)$,  then 
 \begin{align*}
 	\llbracket P \rrbracket_G = & \{ \mu_1 \cup \mu_2| \ \mu_1 \in \llbracket P_1 \rrbracket_G,
 	\mu_2 \in \llbracket P_2 \rrbracket_G, \mu_1 \sim \mu_2 \text{ and }  \mu_1 \cup \mu_2 \models R \}\\
 	&\cup \{\mu_1 \in  \llbracket P_1 \rrbracket_G|\ \nexists \mu_2 \in \llbracket P_2 \rrbracket_G,\mu_1 \sim \mu_2 \text{ and } \mu_1\cup \mu_2\models R\},\end{align*}
%    \todo{Consider adding a footnote / sentence on $\mu_1\cup \mu_2\models R$ being true if there is no R.}
	\item[(4)] If $P$ is $\text{UNION }(P_1, P_2)$, then $\llbracket P \rrbracket_G = \llbracket P_1 \rrbracket_G \cup \llbracket P_2 \rrbracket_G$,
    %\todo{I am not sure the FNE is correct here}
	\item[(5)] If $P$ is $\text{MINUS }(P_1,P_2)$, then $\llbracket P \rrbracket_G = \llbracket P_1 \rrbracket_G \setminus_M \llbracket P_2 \rrbracket_G$,
    \item[(6)] If $P$ is $\text{FE }(P_1,P_2)$, then
   $\llbracket P \rrbracket_G = \{\mu \in \llbracket P_1 \rrbracket_G \ | \ \llbracket P_2\mu \rrbracket_G\neq \emptyset\}$,\footnote{Note that in our definition,  the  term $\llbracket P_2\mu \rrbracket_G$ in the evaluation of FE and FNE could be undefined as $\mu$ could map variables occurring in  $P_2$ to blank nodes which we excluded. In that case we treat the blank nodes as constants. Here, we deviate from
   SPARQL~1.1~\cite{harris2013sparql} to address the known issues related to the \textsc{exists} pattern \cite{sparqlex}.}
  %  $\llbracket P \rrbracket_G = 
   % \{\mu \in \llbracket P_1 \rrbracket_G \ | \ \exists \mu_2 \in \llbracket P_2 \rrbracket_G: \mu \text{ is compatible with } \mu_2 \}, 
   % $ %, where only variables in $P_2$ are in-scope within the context of this operation.
    \item[(7)] If $P$ is $\text{FNE }(P_1, P_2)$, then $\llbracket P \rrbracket_G = \{\mu \in \llbracket P_1 \rrbracket_G \ |\ \llbracket P_2\mu \rrbracket_G=\emptyset\}$,
    %, with $P_2$ variables in-scope only within the local evaluation of this pattern\footnote{Variable scope in SPARQL defines where variables are visible within the query. For example, variables in a \texttt{MINUS} or \texttt{NOT EXISTS} pattern do not propagate outside of these expressions.}.
	%\item[(8)] If $P$ is $(P_1\text{ BIND }exp \text{ AS }x)$, then $\llbracket P \rrbracket_G = \{ extend(\allowbreak \mu, x, exp)  \  |  \  \mu \in \llbracket P \rrbracket_G \}$,
	%\item[(9)] If $P$ is $(\text{VALUE }\vec{V} \ T)$, then $\llbracket P \rrbracket_G = \{\mu \ | \ dom(\mu) = \vec{V} \text{ and } \mu(t) \in T\}$.
    \item[(8)] If $P$ is $\text{FILTER }(P_1,R)$, then $\llbracket P \rrbracket_G = \{\mu \in \llbracket P_1 \rrbracket_G \ | \ \mu \vDash R\}$.
\end{enumerate}
%For the expressions $(P_1 \text{ OPT } P_2)$, $(P_1 \text{ FE } P_2)$, and $(P_1 \text{ FNE } P_2)$ we have special cases if $P_2$ is of the form $(Q \text{ FILTER } R)$. For these expressions all compatible mappings $\mu_1$ and $\mu_2$ involved in the evaluation need  to additionally satisfy the condition $\mu\cup \mu_2 \models R$. For more details we refer to Polleres and Wallner \cite{polleres2013relation}.

%The \textit{extend} function computes the result of an expression $exp$ and applies it to a mapping. For additional details, see Polleres~\cite{polleres2013relation}.
%\todo{refine  Definition}
The domain of a solution mapping is called its scope \cite[18.2.1]{sparqlspec}. %Intuitively, the scope are the variables that are ``relevant to the outside''.
%For this purpose
As the definitions later in this paper will rely on it, we formally introduce it and 
 define the function $sv$ which retrieves the scope of each query pattern.

%\da{I moved a fragment to the SPARQL section from here, I need to adapt the text}
%The evaluation of graph patterns sometimes depends on the scoping of variables. 
Let
$var: GP\rightarrow 2^V$ be a function which maps a graph pattern to the set of variables occurring in it. 
The function $sv: GP\rightarrow 2^V$ % which assigns a scope to each graph pattern 
is defined as follows:
\begin{enumerate}
	\item[(1)] If $P$ is a BGP $t$, then $sv(P)=var(t)$,
	\item[(2)] If $P =X(P_1,P_2)$, with $X \in \{\text{AND}, \text{UNION}\}$ then $sv(P)=sv(P_1)\cup sv(P_2)$,
	\item[(3)] If $P=\text{OPT}(P_1, P_2,R)$, then $sv(P)=sv(P_1)\cup sv(P_2)$
	\item[(4)] If $P=X(P_1 ,P_2)$, with $X \in \{\text{MINUS}, \text{FE},\text{FNE}\}$ then $sv(P)=sv(P_1)$,
	\item[(5)] If $P=\text{FILTER } (P_1 ,R)$, then $sv(P)=sv(P_1)$.
\end{enumerate}

%\da{We need not exists and some text about scoping of variables}
%\da{Question: Is the variable scope (\url{https://www.w3.org/TR/sparql11-query/\#variableScope}) normally already covered in the definition and is that the reason why we do not cover it here since we start from the algebra? For our minus operator that matters, so not sure how to deal with that}
\vspace{0.2cm}
\textbf{Evaluation of SPARQL Queries}\label{sec:sparql_sem}
SPARQL queries normally also come with a query form, that is, a query is written as  %\%wvw{nitpicking; perhaps replace ``in the form'' by ``as'' since ``query form'' is a separate concept}
% SPARQL queries do not only consist of query patterns, but they normally also come with a query form, that is a query is written in the form 
\[
(n~x) \text{ WHERE } P
\]
%\verb| (name x) WHERE Q |.
%returns a multiset of variable mappings from the evaluated graph patterns. Depending on the query form, results may be projected in a tabular format (\textsc{Select} queries) or output as an RDF graph (\textsc{Construct} queries). For other query forms, consult the SPARQL 1.1 specification~\cite{harris2013sparql}.
%
%We assume all queries to have the form
%\verb| (name x) WHERE Q |, that is, we see query forms as a 
where $P$ is a graph pattern and the pair $(n, x)$, called query form, consists of a name
$n \in \{\textsc{select},\textsc{construct}\}$, and argument $x$. The latter is either a list of variables or an  asterisk if $n=\textsc{select}$, or a pattern $Q \subset (I  \cup L \cup V\cup B) \times (I \cup V) \times (I  \cup L \cup V\cup B)$, if $n=\textsc{construct}$.  For all  $\mu \in \llbracket P \rrbracket_G$ the evaluation of \textsc{select} retrieves the variable bindings, the evaluation of \textsc{construct} retrieves instantiated version $Q\mu$ of $Q$,
% For each  $\mu \in \llbracket P \rrbracket_G$ we get
%are provided. If $Q$ contains blank nodes, we get a new instantiation with a "fresh" blank node for each mapping $\mu \in \llbracket P \rrbracket_G$.  
%$Q\mu$, 
where $Q\mu$ stands for the pattern retrieved by replacing each $x$ occurring in $Q$ by $\mu(x)$ and each blank node $b$  by a blank node $b_\mu$ such that $b_\mu$ does not occur in $G$ and $b_\mu\neq b'_\mu$ if $b\neq b'$. %For the remainder of this paper, we make the assumption that all variables and blank nodes in $Q$ are in the 

With the current definition, it can happen, that $\mu$ leaves variables in $Q$ unbound.  This is the case if $Q$ contains "new" variables, that is  if $var(Q)\setminus sv(P)\neq \emptyset$, or if $Q$ contains variables which are instantiated through an \text{OPT} pattern which sometimes does not retrieve a result. For the former case, the SPARQL query does not lead to a result - for the remainder of the paper we exclude this case from our considerations\footnote{Note that we can easily check before evaluation whether $var(Q)\setminus sv(P)= \emptyset$.} - in the latter case, we assume that unbound variables will be replaced by a term indicating their unboudness (in our concrete implementation that will be a special kind of blank node).

\section{Reasoning in N3}
\label{sec:n3}
%
%\todo{Maybe make some special comment about the fact tha N3 allows graph terms and that this makes it possible to follow SPARQL even in structure}
%\todo{built-ins mentioned below (check whether we explain them): includes, not includes}ΩΩΩ
% Having discussed SPARQL and its semantics, we now further introduce Notation3 Logic. 
This section focuses on the aspects of Notation3 Logic relevant for this paper, and refers %the interested reader 
to the full specification provided by the W3C community group \cite{woensel2023n3}.

%N3 directly extends RDF 
%SPARQL Basic Graph Patterns (BGP) by incorporating additional constructs, such as blank nodes and specialized terms known as graph terms and lists.
\subsection{Syntax}
The syntax of SPARQL and N3 is closely related (Section \ref{sec:bground}). In fact, BGPs can be directly used as premises and conclusions of N3 rules, but N3 allows for extra patterns. 
%\wvw{Use of the term ``more complex'' could be up for discussion, I think.}   
% From a syntactical point of view, all SPARQL graph patterns are also valid N3. 
We thus define N3 Graph Pattern (N3GP) as an extension of BGP.
% We take that into account in the definition of N3's syntax and extend the notion of a BGP to an N3 Graph Pattern (N3GP).
%The definiti Arndt et al.~\cite{arndt2023n3semantics} as follows:

\label{N3GP}
Given the sets $I$, $B$, $L$, and $V$, representing IRIs, Blank Nodes, RDF literals, and Variables, respectively, we %recursively 
define the set $T_{N3}:=AT\cup GT\cup LT$ of N3 terms:
\begin{itemize}
\item  $AT:=I\cup B\cup L \cup V$ is the set of \emph{atomic terms},
%Atomic terms in SPARQL are taken from the sets  
%    \item The set $AT$ of \emph{atomic terms}, akin to SPARQL, consists of the following:
%    \begin{itemize}
%        \item $I$: the set of IRIs,
%        \item $L$: the set of literals,
%        \item $V$: the set of variables,
%        \item $B$: the set of blank nodes.
%    \end{itemize}
    \item  $GT:= \{ G \; | \; G \in N3GP \}$ is the set of \emph{graph terms} (N3GP is defined below),
    \item  $LT:=\{( t_1, \ldots, t_n) \; | \; n\geq 0, t_1,\ldots,t_n\in T_{N3}\}$ is  the set of \emph{list terms}.
\end{itemize}
We call a tuple of form $T_{N3} \times T_{N3} \times T_{N3}$ an \emph{N3 triple}. %We call a list which does not contain variables \emph{ground list} and denote the set of ground lists as $LT_g$. 
%We call the $GT_A:=2^{T\times T\times T}$ the set of \emph{atomic graph terms} (see Section  \ref{sec:sparql} for the definition of $T$).
%Accordingly, we call the set of graph terms not contining variables the set of \emph{ground graph terms} and denote it by $GT_g$
%We call $T_c:=T\cup LT_g\cup GT_g$  the set of \emph{closed N3 terms}. Accordingly,
%tuples of the form $T_c\times T_c\times T_c$ are called \emph{closed N3 triples} and a set consisting of closed triples a \emph{closed N3 graph}.
An \emph{N3 graph pattern} (N3GP) is a set of N3 triples. Note that $BGP\subset N3GP$. 

% \da{Maybe add: In N3 graph terms, blank nodes are scoped locally, that is, two isomorphic graph terms 
% are interpreted equally. To make that possible, we introduce the notion of closed terms.} \dt{+1}

We call graph terms \emph{closed} if they, and their possibly nested subgraphs, do not contain variables, and denote their set by $GT_c$. % the set of closed graph terms by $GT_c$.  
If lists, and their possible nested sublists, only contain elements of $I\cup L\cup GT_c$, we call them \emph{closed}. We denote the set of closed lists by  $LT_c$ and define the set of closed terms as $T_c:=I\cup L\cup GT_c\cup L_c$. We call a set of \emph{closed} triples  $G\subset T_c\times T_c\times T_c$ a $\emph{closed}$ graph.
%\dt{a single sentence explaining why closedness matters} \da{I did not mention it before because the paper is not about N3 semantics but about our translation. I now added it but am afraid that it needs more explanation which we unfortunately don't have. }\dt{OK, I like both options but at least we should to something}\wvw{I personally don't think it's a problem since we rely on closedness in the semantics section.}
%We call lists  $(t_1,\ldots,t_n)$ only containing elements $t_i\in I\cup L\cup GT_c$ closed lists and denote the set of closed lists as $LT_c$.  Accordingly, we introduce \emph{closed terms} $T_c=$The set of closed graphs is the set $2^()$

\subsection{Semantics}\label{sec:N3semantics}
As the syntax reveals, N3 is rather expressive as it supports lists and graph terms. 
We present a simplified version of its semantics, focusing on the aspects relevant to this work (see the Community Group Report \cite{arndt2023n3semantics} for more details).

%\vspace{0.15cm}
\noindent \subsubsection{Base semantics}
N3 semantics extends RDF simple entailment \cite{RDFsemantics}.  
Given a closed N3 graph, N3's basic interpretation function $\mathfrak{I}$ maps closed 
terms into the domain of discourse. The domain of discourse contains normalised syntactic representations of graph terms and canonical representations of lists, and $\mathfrak{I}$ maps graph terms and lists to their representations. For all other ground terms, it corresponds to RDF's interpretation function $I$. Relations are interpreted using RDF's function $IEXT$ and a triple $\langle t_1, t_2, t_3\rangle$  is true iff $(\mathfrak{I}(t_1),\mathfrak{I}(t_3))\in IEXT(\mathfrak{I}(t_2)) $. A closed graph is true iff all its triples are true. 

%\todo{something more specific than ``more general''?} 
Closed triples are those which are not (implicitely) quantified on the graph level. They are the N3 counterpart of the RDF graphs containing blank nodes. To interpret non-closed N3 graph patterns, we rely on mappings for blank nodes and variables. 
%\todo{best way to start here?}
%To properly interpret graph terms, N3 mappings are 
%% 
In N3, these mappings can be seen as 
a mixture of RDF's mapping function $A$, which maps blank nodes to the domain of discourse, and the variable mappings found in the SPARQL semantics (Section~\ref{sec:sparql_sem}). 
We use mappings of the form $A: X\rightarrow D\times T_c$, where $X\subseteq B$ or $X\subseteq V$, and  $D$ is the domain of discourse and $T_c$ as defined before. 
% (wvw) separating into two sentences
For each $A(x)=(A_1(x),A_2(x))$, we have $\mathfrak{I}(A_2(x))=A_1(x)$.
%variables and blank nodes in N3 are mapped to a pair of a closed term and its interpretation. 
If a variable is found in a graph term, or any of its nested subgraphs, then it is first grounded using $A_2(x)$ and then the resulting closed graph is interpreted.
%Then, we can safely interpret the graph term by mapping it to its .
Otherwise, $A_1$ is used to directly map the variable into the domain of discourse.
% To interpret an N3 graph pattern, we apply $A$ depending on its context. 
% If a variable $x$ occurs in a graph term or a list, we ground it using $A_2(x)$; otherwise, we use $A_1(x)$ to interpret an N3 graph pattern.
Variables are universally quantified and blank nodes are existentially quantified. With some exceptions we will discuss later the scope of variables is global, the scope of blank nodes is local to the graph term. If a triple or graph $g$ is true under $\mathfrak{I}$, we also write $\mathfrak{I}\models_b g$.
%\todo{add a part on blank nodes} 
We further write $\mathfrak{I}[A]$ for the interpretation that uses a mapping $A$ to interpret variables or blank nodes. 

 %. This predicate does not carry any specific meaning in \emph{base semantics} 
%\checkmark \wvw{it's an interesting example :-) but, for space considerations, do we absolutely need it (aside from being a segue to the next paragraph)?} 
%\da{The problem here is that the way of defining entailment is very unusual and I give the example to make sure the reviewers understand the need for this complexity. So, we can leave it out, but I am afraid of complains in that regard... Maybe you have abother suggestion? Shorter example?}
%\wvw{I felt bad about taking up more space above so that was my main motivtion. Ok if you think we should keep it!}\da{Not sure yet, maybe we come up with something better?}\wvw{Commented out the example for now}

 % N3 rule~\ref{n3rule} 
 %, we know that
 \sloppy
%Last, we give an example. 
%As an example, consider rule~\ref{n3rule}. The implication symbol $\rightarrow$ (\texttt{=>} in concrete syntax) is syntactic sugar for the predicate 
%\texttt{log:implies} which we abbreviate by "$\impl$":
 % $\langle \langle ?x, \type, \Researcher\rangle, \impl,   \langle ?x, \type, \Person\rangle \rangle$.  For base semantics, the formula forms a triple with graph terms, which is  true iff for all mappings $A$ as above, $( \mathfrak{I}(\langle A_2(?x), \type, \Researcher\rangle), \mathfrak{I}(\langle A_2(?x), \type, \Person\rangle)) \in IEXT(\mathfrak{I}(\impl)).$
%N3 encompasses all RDF elements, including literals, IRIs, blank nodes, and variables, making all BGPs valid N3GPs. SPARQL variables and N3 universals are similar, both serving as placeholders bound during query evaluation or logical inference. N3 distinguishes universally scoped variables from existentially scoped variables.
% \da{I added the example again, but it can go when we  shorten the paper}
 \sloppy
%Last, we give an example. 
As an example, consider rule~\ref{n3rule}. The implication symbol \texttt{=>} is syntactic sugar for the predicate 
\texttt{log:implies} which we abbreviate by "$\impl$":
$\langle \langle ?x, \type, \Researcher\rangle, \impl,   \langle ?x, \type, \Person\rangle \rangle$.  For base semantics, the formula forms a triple with graph terms, which is  true iff for all mappings $A$ as above, $( \mathfrak{I}(\langle A_2(?x), \type, \Researcher\rangle), \mathfrak{I}(\langle A_2(?x), \type, \Person\rangle)) \in IEXT(\mathfrak{I}(\impl)).$

\subsubsection{Log Semantics}\label{builtin}
%As exemplified above with the predicate $\impl$, 
Logical operations like $\impl$ in N3 are expressed by built-in predicates, whose special meaning is not covered by the base semantics. For the logical built-in predicates which we collect in the set $LP$ we define the log semantics as an extension of base semantics.
%For instance, the implication symbol \texttt{=>} is syntactic sugar for the predicate 
%\texttt{log:implies}.
%The meaning of such \emph{logical} predicates, the set $LP$, is instead given by the log semantics. 

% (wvw) dividing up these sentences
We say that an interpretation $\mathfrak{I}$ is log model for a graph  $g$, written as $\mathfrak{I}\models_l g$, %\todo{just in general, if or iff?} 
iff $\mathfrak{I}\models_b g$ and for each triple $\langle s, p, o\rangle\in g$ with $p\in LP$, the triple fulfills extra conditions which depend on the particular $p$.  
We give these conditions for the predicate $\impl$. 
$\mathfrak{I}\models_l \langle s, \impl, o \rangle$ 
iff 
%$\langle s, \impl, o \rangle$ is true under log semantics if 
  %there exists a 
  %mapping 
  %\todo{how to explain blank nodes here? is it necessary?}
  % with
%it holds that 
(\textbf{1})
$\mathfrak{I}\models_b\langle s, \impl, o \rangle$, and 
%(2) for each model 
(\textbf{2}) if for $A:V\rightarrow D\times T_c $ 
it holds that $\mathfrak{I}[A]\models_b \langle s, \impl, o \rangle$ %for which 
and $\mathfrak{I}[A](s),\mathfrak{I}[A](o)\in GT$ then $\mathfrak{I}[A]\models_l o$ if  $\mathfrak{I}[A]\models_l s$.

To  illustrate the definition, we come back to the interpretation of rule~\ref{n3rule}
and show how we can use it to derive triple \ref{john2} from triple \ref{john}.
%have in mind that $\mathfrak{I}_1(s)$ and $\mathfrak{I}_1(o)$ are \emph{syntactical} objects. To come back to our previous example on  , % 
%
%$\langle s, \impl, o \rangle$ is true under base semantics. 
%
%If $(\mathfrak{I}(s), \mathfrak{I}(o))\in IEXT(\mathfrak{I}(\impl))$ and $\mathfrak{I}(s),\mathfrak{I}(o)\in GT$ then for each  interpretation $\mathfrak{I}_2$ it holds that $\mathfrak{I}_2\models \mathfrak{I}(o)$ if  $\mathfrak{I}_2\models \mathfrak{I}(s)$
%
%
%\begin{itemize}
%	\item $\mathfrak{I}\models \langle s, \impl, o \rangle$
%	\item if $\mathfrak{I}(s)$ and $\mathfrak{I}(o)$ are graph terms
%\end{itemize}
%To illustrate the idea on the example above, 
%let's assume that 
Let us assume that $\mathfrak{I}\models_l \{\langle john, \type, \Researcher\rangle, \langle \langle x, \type, \Researcher\rangle, \impl,   \langle x, \type, \Person\rangle \rangle\}$ and that $\mathfrak{I}(john)=j$. 
Then %That means that 
%As the implication is true under 
%$\mathfrak{I}$, as per the above: \newline
%\hspace*{0.15cm} (\textbf{1}) 
$\mathfrak{I}[A]\models_b \langle \langle x, \type, \Researcher\rangle, \impl,   \langle x, \type, \Person\rangle \rangle$ holds for all
mappings $A$ (according to (1)), % (wvw) just to reduce the number of symbols here (they were hurting my head) would just "mapping" be enough?
%$A:V\rightarrow D\times T_c$ 
%with $\mathfrak{I}(A_2(x))=A_1(x)$,
and thus also for $A'$ such that
%that $\mathfrak{I}[A]\models \langle \langle x, \type, \Researcher\rangle, \impl,   \langle x, \type, \Person\rangle \rangle$,
$A'(x)=(j, \text{john})$. 
% \newline
%\hspace*{0.15cm} (\textbf{2}) 
%The subject and object are graph terms, so, for all $A$, 
As $\mathfrak{I}[A'](\langle x, \type, \Researcher\rangle), \mathfrak{I}(A')(\langle x, \type, \Person\rangle) \in GT$ it furthermore holds that if 
$\mathfrak{I}[A']\models_l \langle x, \type, \Researcher\rangle$ then $\mathfrak{I}[A']\models_l \langle x, \type, \Person\rangle$ (according to (2)). As %we know that
$\mathfrak{I}\models_l \langle john, \type, \Researcher \rangle$
%if $\mathfrak{I}[A]\models_l \langle x, \type, \Researcher\rangle$,  we can infer $\mathfrak{I}[A]\models_l \langle x, \type, \Person\rangle$. As $\mathfrak{I}\models_l \langle john, \type, \Researcher\rangle$ 
and $A_1'(x)=\mathfrak{I}(A_2'(x))=\mathfrak{I}(john)=j$, we get that
$\mathfrak{I}\models_l \langle john, \type, \Person\rangle$.

%we can decide whether or not we want it
%\input{example}
 
% we also know that triple \ref{john2} holds. 
%with  $\mathfrak{I}(john)=j$ as an example.
%and that $( \langle \text{john}, \type, \Researcher\rangle, \langle \text{john}, \type, \Person\rangle) \in IEXT(\mathfrak{I}(\impl)).$
%To be true under $\mathfrak{I}$ under log semantics, it furthermore needs to hold that 
%$\mathfrak{I}[A]\models \langle x, \type, \Person\rangle$ if $\mathfrak{I}[A]\models  \langle x, \type, \Researcher\rangle$. In our example, that includes that is 
%\[\langle \langle \text{john}, \type, \Researcher\rangle, \impl,   \langle \text{john}, \type, \Person\rangle \rangle.\]
%then, this triple is true according to the built-in semantics iff
%\[ \mathfrak{I}\models \langle \langle \text{john}, \type, \Researcher\rangle \text{ implies that }  \mathfrak{I}\models \langle \text{john}, \type, \Person\rangle \rangle.\]

\sloppy
The semantics of the other built-in predicates is defined in a similar way. 
We only list those which are relevant for this work.
%For brevity, we do not list their formal interpretations in this paper and provide a short explanation instead. 
%In our paper, only
For the predicate \texttt{log:includes} (or ``$\includes$'' for brevity) 
%and its counterpart \texttt{log:notIncludes} (or ``$\notIncludes$'' for brevity)
we have that $\mathfrak{I}\models_l\langle s, \includes, o \rangle$ if  $\mathfrak{I}\models_b\langle s, \includes, o \rangle$ and if  for $A:V\rightarrow D\times T_c $  it holds that $\mathfrak{I}[A]\models_b\langle s, \includes, o \rangle$
%for all for all mappings $A$ with %$\mathfrak{I}[A]\models_b\langle s, \includes, o \rangle$ %$\mathfrak{I}(s),\mathfrak{I}$
%and 
%compares two N3GP and a triple $\langle s, \includes, o\rangle$ is true iff  
and 
$\mathfrak{I}[A](s),\mathfrak{I}[A](o)\in GT$ 
then there exists a mapping  $\mu$ from the blank nodes occurring in  $o$  to $T_C\cup B$ such that  $\mathfrak{I}[A](o)\mu \subseteq \mathfrak{I}[A](s)$.
%then there exists a mapping $A^s$ 
% $\mathfrak{I}[A](s)$ and $\mathfrak{I}[A](o)$ that there exist mappings $A^s$ and $A^o$ such that $\mathfrak{I}[A][A^{o}](o)\subset \mathfrak{I}[A][A^{s}](s)$.
%$\mathfrak{I}[A](s),\mathfrak{I}[A](o)\in N3GP$ then there exists a graph isomorphism $i$ such that $i(\mathfrak{I}[A](o))\subseteq \mathfrak{I}[A](s)$. 
The predicate has one peculiarity: if $\langle s, \includes, o \rangle$ occurs in an asserted triple and $s$ or $o$ are graph terms containing variables, then these variables are treated as blank nodes. If  the triple occurs in 
 the premise of a rule and $s$ or $o$ are graph terms containing variables which do not occur anywhere else in the same rule, then these variables are  as well treated as locally existentially quantified. The base interpretation will map the graph term containing these local variables to a normalised representation where these are replaced by blank nodes. 
For $\includes$, we have a special case if the subject of the triple is not specified.  
 If we write $\langle \_, \includes, o\rangle$\footnote{In concrete syntax, we would use an arbitrary variable.}, subject gets instantiated by %for the built-in interpretation 
 %$\mathfrak{I}_b$ we have
the deductive closure $close(F)$ of the input graph $F$, that is by a graph term containing all triples which can be derive from $F$ (note that in N3 the input graph also includes the rules). %$F \models_l o\mu$ 
$ \mathfrak{I}\models_l \langle \_, \includes, o\rangle$
if  $\mathfrak{I}\models_b\langle close(F), \includes, o \rangle$ and if  for $A:V\rightarrow D\times T_c $  it holds that $\mathfrak{I}[A]\models_b\langle close(F), \includes, o \rangle$
%for all for all mappings $A$ with %$\mathfrak{I}[A]\models_b\langle s, \includes, o \rangle$ %$\mathfrak{I}(s),\mathfrak{I}$
%and 
%compares two N3GP and a triple $\langle s, \includes, o\rangle$ is true iff  
and 
$\mathfrak{I}[A](o)\in GT$ 
then there exists a mapping  $\mu$ from the blank nodes occurring in  $o$  to $T_C\cup B$ such that  $\mathfrak{I}[A](o)\mu \subseteq close(F)$.\footnote{Note that N3 does normally not allow for infinite graph terms.}
%$\mu$ from the blank nodes and local variables in 
%$\mathfrak{I}(o)\subseteq close(F)$. 
In short that means that $\mathfrak{I}[A](o)$ can be derived from $F$.
%for some 
%mapping $\mu$ from the blank nodes in $o$ to $T_C\cup B$, that is, we unify the anonymous blank node with the deductive closure of the input graph, and not with all graph terms in the domain of discourse. 
 %The predicate \texttt{log:notIncludes} (``$\notIncludes$'' for brevity) is the negation of $\includes$ and follows the same rules for local scoping. We have
 %, but it has one peculiarity: if $\langle s, \notIncludes, o \rangle$ occurs in the premise of a rule and $s$ or $o$ are graph terms and contain variables which do not occur anywhere else in the rule, then these variables are treated as locally existentially quantified. The base interpretation will map the graph term containing these local variables to a normalised representation where these are replaced by blank nodes.  Assuming that the interpretation 
 %$\mathfrak{I}$ applies this special treatment we get
  $\mathfrak{I}\models_l \langle s, \notIncludes, o \rangle$ iff $\mathfrak{I}\not\models_l \langle s, \includes, o \rangle$.
N3 furthermore contains predicates to produce a copy of a graph term with renamed variables (\texttt{log:copy}, "$\cop$" for brevity), e.g., $\{\langle ?x, p, o \rangle\}$ becomes $\{\langle ?x_{new} , p, o \rangle\}$ and to produce a conjunction, i.e. the union, of two graphs (\texttt{log:conjunction}, "$\conj$" for brevity).
%for term equality (\texttt{log:equalTo}) and inequality (\texttt{log:notEqualTo}).
%these operate on a syntactical level and compare the representation of terms.
%\todo{most likely, I will put the locality to includes}

\subsection{Forward and backward reasoning}
\label{sec:fvb}
% In addition to the semantics as provided above, 
%\dt{a bridging sentence here}
Reasoners implement the semantics presented above using different execution strategies; %to derive new knowledge and 
some reasoners (notably EYE \cite{eyepaper}) even enable the user to specify how rules will be applied.
%The Eye reasoner~\cite{eyepaper}, for example, allows specifying how rule reasoning should be executed; i.e., 
% Here, we choose between forward and backward reasoning. 
% (wv) slightly rewrote this part to match the evaluation
Forward reasoning presents a bottom-up approach: given a set of initial facts, the reasoner applies rules to infer new conclusions; these are then used as extra input into the rules, leading to more conclusions. 
This process continues until no other conclusions can be drawn.
% This process involves firing any rule where the premise holds, and adding the inferences to the set, continuing until no more rules can be applied. 
Backward reasoning is a top-down approach that begins with a given query, such as {\small \verb|?x a :Person|}. The reasoner searches for rules with conclusions that may satisfy the query, and then checks whether the premises of those rules hold, by 
searching for matching facts or other rules with matching conclusions. Here, the process ends if all premises can be matched to facts.
% Forward reasoning thus involves materializing the entire state space, i.e., generating all inferences. %(e.g., including transitive closures).
% Backward reasoning is able to search only part of the state space. 
While the execution strategies clearly differ, 
%with forward reasoning progressing from initial information and backward reasoning working in reverse to substantiate a goal, 
both can be considered equivalent in terms of expressivity.
Our evaluation illustrates how backward reasoning can have a distinct benefit in case of large search spaces. %, as it is able to search only part of it.
% show the practical utility of both evaluation strategies, and how it depends on the scenario at hand.

% Forward reasoning presents a bottom-up approach: starting from a set of initial and inferred statements, a reasoner systematically applies known facts and rules using modus ponens to derive new conclusions step by step. This process involves firing any N3 rule where the premise holds and adding the inferences to the set, continuing until no more rules can be applied. In contrast, backward reasoning is a top-down approach that begins with a specific query or goal, such as \verb|?x :locomotion :flying|. The reasoner searches for rules with conclusions that may satisfy the query and then checks whether the premises of those rules hold. This validation process may require searching for other rules with matching conclusions. While the operational strategies differ, with forward reasoning progressing from initial information and backward reasoning working in reverse to substantiate a goal, the two modes can be considered equivalent in terms of general expressivity.
\section{From SPARQL to N3}
\label{sec:sparql2n3}
In this section, we define the actual mapping between SPARQL and N3 following the idea presented in Section~\ref{sec:bground}. 
Provided with a query 
\[
(n~x) \text{ WHERE } P \] 
we define a mapping $m$ from SPARQL graph patterns to N3 graph patterns (Section \ref{sec:premise}) and a mapping $h$ from the whole query to N3 graph patterns (Section \ref{sec:consequence}) such that our translated rule is
\[
m(P)\rightarrow h((n~x), P)
\]
With this setup, we ensure that each SPARQL query is represented by one single N3 rule. 
However, this requires the direct translation of SPARQL concepts to equivalent N3 constructs, which cannot be done for \textsc{union} and \textsc{optional}. 
To implement these concepts as N3-native constructs, we require an extra ``runtime'' set of rules.
% As some SPARQL concepts, namely \texttt{UNION} and \texttt{OPTIONAL}, are not directly supported by N3, our approach requires extra rules which handle these concepts at run time --- our ``runtime'' rules. 
That is, the translated queries $m(P)\rightarrow h((n~x), P)$ always need to be combined with a fixed set of rules. We explain these rules in Section~\ref{sec:n3_special}. % The proofs for the lemmas provided below can be found in our online appendix \cite{extend}.
 
% We thus separate the construction of the rule premise from the consequence and treat each of them in its own section. As a third part, we discuss 
%the predicates covering different functionalities present in SPARQL but not in N3, such as the optional operator or property paths. We opted to model these through special predicates, for which we also provide the N3 rules. The advantage of this approach is that it enables N3 users coming from SPARQL to directly use the constructs they are familiar with. These rules are used in our application, but can also be reused independently in different contexts.

%We have seen above, that BGPs are a subset of N3GPs and that these two are also semantically related. 
%As a consequence of that, we can translate simple WHERE-clauses only consisting of BGPs directly to the premises of N3 rules. The query

\subsection{From SPARQL graph patterns to N3 premises}\label{sec:premise}
%\da{plan: I make the scoping story simpler and just say that we need the domains of the solution mappings and these are required for the minus (as the semantics says). I plan to even get rid of the example, but maybe we still need it to explain the renaming of the filter}
%\todo{I modified this definition of open and closed constructs, we need to double check with the SPARQL section}
 %This section aims to provide a translation function that transforms the SPARQL graph pattern of a query into an N3 graph pattern, which can then be used as the premise of the translated N3 rule.  
% As discussed above  (see Section \ref{sec:varscope}), some constructs like the Kleene-star in the \textsc{select} query form or the \textsc{minus} operator in a query pattern take variable scoping into account. 
Before proceeding, %coming to the actual translation from the SPARQL query to the rule premise, 
we need to define an auxiliary function to cope  with the MINUS operator.  
In the SPARQL semantics (Section \ref{sec:sparqlsemantics}) this is evaluated as the mapping-difference between the sets of mappings $\Omega_1 \setminus_M \Omega_2 = \{\mu \in \Omega_1 \mid \text{for all } \mu^\prime \in \Omega_2: \mu \not\sim\mu^\prime\text{ or } \text{dom}(\mu)\cap \text{dom}(\mu^\prime)=\emptyset  \}$. For our translation, that means that we have to cover two cases: (M1) the solution mappings are not compatible and (M2) the domains of the solution mappings are disjoint. 
% (wv) possible candidate for removal I think
% The second condition can be easily checked without even knowing the data graph as SPARQL semantics clearly defines which variable bindings  each (sub-)query retrieves.  
%\todo{adjust!}

%\da{this is where I stopped, reason: I wonder how to explain the need of renaming. Most likely, I will put the example back in. Question: do I also need renaming for FE and FNE?}
%Renaming is needed to handle filter not exists nested in minus

Regarding MINUS, case (M2) can be easily checked without even knowing the data graph; given $\text{MINUS}(P_1,P_2)$, if the scopes $sv(P_1)$ and $sv(P_2)$ are disjoint, then we can simply proceed with $P_1$.  If they are not disjoint, we need to distinguish between shared and local variables (case (M1)).

% Our translation uses this function to determine \wvw{second part is unclear} whether a MINUS pattern needs to take its second argument into account.  
% If this is the case (case (1) from above), \wvw{unclear to me at this point (we haven't said that we wouldn't rely on scope)} we still rely on the scope in our translation. 
% \wvw{you already mention this later on so I think it can be left out here} 
% To  explain that we emphasize that in our translation one SPARQL query is always represented by one rule. Having that in mind, we consider an example. 

\label{sec:sparql-n3}
\begin{lstlisting}[caption={Example query demonstrating the behavior of variable scoping for \textsc{minus} with a nested \textsc{filter exists}.}, label={lst:minus-example}, float=t]
SELECT * WHERE { 
  ?x :p ?n . MINUS { ?x :q ?m . FILTER EXISTS {?m :r ?n}}}
\end{lstlisting}

To illustrate the complexity of MINUS, we provide an example. 
In Listing~\ref{lst:minus-example}  we display a MINUS query with a nested FILTER EXISTS query in its second argument. In our notation from above, the query pattern can be written as $P=\text{MINUS}(\{\langle?x, p, ?n\rangle\}, \text{FE}(\{\langle ?x,q,?m\rangle,\langle ?m, r, ?n\rangle\}))$.
%In scope of the first argument are the variables \texttt{?x} and \texttt{?n}, in scope of the second argument are the variables \texttt{?x} and \texttt{?m} and as \texttt{?x} occurs in both scopes, we are clearly in case (1) from above. 
As the domains of $P_1$ and $P_2$ are not disjoint, case (M2) from above is not relevant.
To check case (M1), we evaluate the query on graph 
$G= \{\langle s, p, o\rangle,\langle s, q, a  \rangle, \langle a, r, b \rangle\}$.
%\begin{equation}
%	\texttt{:s :p :o. :s :q :a. :a :r :b.}
%\end{equation}
The first argument $P_1= \{\langle?x,p,?n  \rangle\}$ %\texttt{\{?x :p ?n.\}}$ 
of MINUS matches the first triple from $G$ and we get the solution mapping $\llbracket P_1 \rrbracket_G=\{\{(?x, s), (?n, o) \}\}$. To evaluate the second argument, we first evaluate $P_{2,1}=\{\langle ?x, q, ?m\rangle\}$ and retrieve $\llbracket P_{2,2} \rrbracket_G=\{\{(?x, s), (?m, a) \}\}$. 
Then, we apply the mapping $\mu=\{(?x, s), (?m, a) \}$ 
on $P_{2,2}=\{\langle ?m, r, ?n \rangle\}$, i.e., $P_{2,2}\mu=\{\langle a, r, ?n \rangle\}$. 
% (wvw) added (& later on a bit modified)
Next, when evaluating $P_{2,2\mu}$ on the data, triple $\langle a, r, b\rangle$ allows us to unify $?n$ with $b$, yielding a non-empty solution mapping for this expression and thus
$\llbracket P_2 \rrbracket_G=\{\{(?x,s), (?m,a) \}\}$. As the solution mappings of $P_1$ and $P_2$ are compatible, we get $\llbracket P_{2,2} \rrbracket_G=\emptyset$.
	%$\mu=\{(\texttt{?x}, \texttt{:x}),(\texttt{?m}, \texttt{:m}), (\texttt{?m}, \texttt{:m}) \}$ as a soltuion mapping for the overall query. 

We observe that, during evaluation, two different values were assigned to the variable $?n$: once with $o$ while evaluating $P_1$, and once with $b$ while evaluating $P_{2,2}\mu$.
This is not a problem in SPARQL, as no actual binding of variables occurs in the second argument of FE (here, $P_{2,2}\mu$), and there is thus no conflict.
However, we aim to translate one SPARQL query to one single N3 rule, and N3 does not allow two different bindings for the same variable in the same premise (the scope of variables is global in N3; Section \ref{sec:N3semantics}). We thus have to relabel the second occurrence of variable $?n$. 
%
%If two variables with the exact same name occur in the premise of an N3 rule, they get the same binding twice as the scope of universal variables is global in N3 . 
 %As in our approach we always translate one query by one single rule, we need to relabel the second occurence of the variable \texttt{?n}. 
%
%we get $\llbracket P_2 \rrbracket_G=\{\mu \in \llbracket \{\texttt{}\} \rrbracket_G\}$
%
%which retrieves instances of nodes (\texttt{?x} and \texttt{?n}) connceted through the predicate \texttt{:p} for which first node (\texttt{?x}, we are  thus in (1)) is not connected through the predicate \texttt{:q} to another node (\texttt{?m}) which itself connected via prodicate \texttt{:r} to \emph{any other} node (here \texttt{?n} again). Note that the variable \texttt{?n} occurs twice, but if we evaluate the query, it is possible, that these two occurences of \texttt{?n} retrieves different bindings.
%
%We also use  this function for the case that the two domains are not disjoint and we need to handle 
%
% For cases where variable scoping is important, our translation relies on a 
We do that with the function $rl: N3GP \times 2^V\rightarrow N3GP$ that relabels all variables found in the first argument, \textit{except} those given as the second argument, by ``fresh'' variables that neither occur in the first nor second argument.
% We do that by the function $rl: N3GP \times 2^V\rightarrow N3GP$ which takes an N3 graph pattern and a set of variables into account and relabels all variables occurring in the input graph except those given as the second argument by "fresh" variables, that is, by variables that neither occur in the pattern nor in the set of variables. 
An example relabeling for $P_{2,2}$ is 
% For pattern $P_{2,2}$ from above we get
%$
%P'=\{\langle ?x, q, ?m\rangle, \langle ?n, \text{equalTo}, ?m\rangle\}
%$
%which is the direct N3-translation of pattern our \textsc{filter} pattern above, we get
$
rl(P_{2,2},\{?x, ?m\})=\{\langle ?m, r, ?n_{new}\rangle \}
$. Note that this relabeling does not influence the result of the query pattern as $P'=\text{MINUS}(\{\langle?x, p, ?n\rangle\}, \text{FE}(\{\langle ?x,q,?m\},\{?m, r, ?n_{new}\}))$ retrieves the exact same bindings as $P$.

We use scoping and relabeling functions to define our translation function from SPARQL query patterns to N3 graph patterns.
% Here, we distinguish between \emph{open} and \emph{closed} constructs. Whether a construct is considered to be open or closed depends on its scoping (see Section \ref{sec:varscope}). We call the operators \texttt{OPT}, \texttt{UNION}, and \texttt{AND} \emph{open} as they do not restrict the scope.
% We call
%\texttt{MINUS}, \texttt{FNE}, \texttt{FE}, and \texttt{F} \emph{closed} constructs as these close or limit the scope in some way. As the closed constructs depend on functions dealing with the scoping, we discuss these in a separate paragraph. First, we discuss the translation of open constructs.
%
%
Let $Q$ be a graph pattern and let $ft$ a function which translates a filter expression to its N3GP translation. We define the premise mapping $m:GP\rightarrow N3P$ recursively as follows:
\begin{itemize}
\item if $Q$ is a BGP $P$, then $m(Q)= P$,
\item if $Q$ is $\text{AND}(Q_1, Q_2)$, then 
$m(Q)= m(Q_1)\cup m(Q_2)$,
 \item if $Q$ is $\text{OPT}(Q_1,Q_2,R)$, then $m(Q)=\{ \langle m(Q_1), \text{opt}, (m(Q_2)\cup ft(R))\rangle \}$,
\item if $Q$ is $\text{UNION}(Q_1,Q_2)$, then $m(Q)=\{ \langle m(Q_1), \text{union}, m(Q_2)\rangle \}$,
% 1/ create rule with P1 as triple pattern, together with constraint that scope should not include P2 (relabeled), if P1 & P2 share variables
 %\cup \{P_1|P_1\in m(Q_1) \wedge   P_2\in m(Q_2)\wedge
% sv(P_2)\cap sv(P_1)= \emptyset\}
%\end{itemize}
%Let $Q$ be a graph pattern. We define the premise mapping $m:GP\rightarrow 2^{N3P}$ recursively as follows:
%\begin{itemize}
%\item if $Q$ is a BGP $P$, then $m(Q)=\{ P\}$,
%\item if $Q$ is of the form $Q_1 \text { AND } Q_2$, then 
%$m(Q)=\{P_{1}\cup P_{2}| P_{1}\in m(Q_1)\wedge P_{2}\in m(Q_2) \}$
% \item if $Q$ is of the form $Q_1 \text{ OPT } Q_2$, then $m(Q)=\{ \{ \langle P_1, \text{ opt }, P_2\rangle \}| P_1\in m(Q_1)\wedge P_2\in m(Q_2)\}$
%\item if $Q$ is of the form $Q_1 \text{ UNION } Q_2$, then $m(Q)=m(Q_1)\cup m(Q_2)$
% 1/ create rule with P1 as triple pattern, together with constraint that scope should not include P2 (relabeled), if P1 & P2 share variables
\item if $Q$ is $\text{MINUS}(Q_1,Q_2)$, 
then 
\[m(Q)=\begin{cases}
 m(Q_1), \text{ if } sv(Q_2)\cap sv(Q_1)= \emptyset\\
 m(Q_1)\cup \{\langle \_, \notIncludes, rl(m(Q_2), (sv(Q_2)\cap sv(Q_1) ))\rangle 
 \}, \text{ else},
 \end{cases}
 %\cup \{P_1|P_1\in m(Q_1) \wedge   P_2\in m(Q_2)\wedge
% sv(P_2)\cap sv(P_1)= \emptyset\}
 \]
\item if $Q$ is $\text{FNE}(Q_1, Q_2)$, then $m(Q)=m(Q_1) \cup \{\langle \_, \notIncludes, m(Q_2)\rangle\}$,
\item if $Q$ is $\text{FE}(Q_1,Q_2)$, then $m(Q)= m(Q_1)\cup \{\langle \_, \includes, m(Q_2)\rangle\}$,
\item if $Q$ is $\text{FILTER} (Q_1, R)$, then $m(Q)=m(Q_1)\cup ft(R)$.
% where $ft$ is a function which translates the filter expression to the corresponding N3GP.
%\item if $Q$ is of the form  $Q_1 \text{ BIND } E \text{ as } x$
\end{itemize}

We do not detail the function $ft$ provided in the above definition because the exact translation depends on the filter element. In most cases, there exist N3 built-in functions that behave identically to SPARQL's filter functions. E.g., the filter expression \texttt{?m < ?n} becomes the N3 triple \texttt{?m math:lessThan ?n.}

% (wv) just shortening (don't think we need the example)
Note that the translation simply converts \textsc{union} and \textsc{optional} patterns to corresponding N3 triples with custom ``union'' and ``optional'' predicates.
% Note that the translation leaves \textsc{union} and \textsc{optional} patterns untouched. The query pattern
% $ \text{UNION}(\{\langle ?x, p, ?y\rangle \}, \{ \langle ?x,q, ?z\rangle\})$
% simply becomes 
% $ \langle\{\langle ?x, p, ?y\rangle \}, \text{union}, \{\langle?x, q, ?z\rangle\}\rangle$ where union is an N3 predicate we introduce  for this translation. 
% (wv) we had already mentioned that so could leave it out
% This is because it is not possible to handle  \textsc{union} and \textsc{optional} by one single rule. 
For these two predicates, we define our run-time rules 
%We discuss the rules to further handle these constructs below 
(Section~\ref{sec:n3_special}).

For the \textsc{minus}, note that we relabel all variables in $Q_2$, \textit{except} for those that occur in the scope of $Q_1$ and $Q_2$. Intuitively, we thus only relabel variables that are \textit{not} relevant to the outside, and should thus not conflict with ``outside'' variables.
Coming back to our example from Listing~\ref{lst:minus-example}, 
the variables excluded from relabeling are $sv(\langle ?x, p, ?n\rangle) \cap 
sv(\text{FE}(\{\langle ?x, p, ?m\rangle, \langle ?x, q, ?n\rangle\})) = \{?x, ?p\}\cap\{?x, ?m\}=\{?x\}$, leading to the application of the relabeling function $rl(m(Q_2), \{?x\})$, i.e., all variables except $?x$  are relabeled. This ultimately leads to
$m(P)= \{ \langle ?x, p, ?n \rangle, \langle \_ ,\notIncludes ,\{\langle ?x, q, ?m_{new}\rangle, \langle \_, \includes \langle ?m_{new}, r, ?n_{new}\rangle \rangle \}\rangle\}$, which solves the issue of multiple bindings per variable in N3.

It still remains to show that the mapping we defined works correctly. We do that for RDF ground graphs which do not contain any built-in IRI of N3. We call these graphs proper RDF graphs,  We limit our consideration to ground graphs to avoid extra difficulties with N3's local scoping. Otherwise we could not consider the rule premise in isolation.
As the translation function for the filter is out of scope in this paper and the runtime rules for the predicates "opt" and "union" will only be discussed in Section \ref{sec:n3_special},  we rely on the following assumption:

\textbf{Assumption} Given an RDF ground graph $G$ for each mapping $\mu:X\rightarrow  I \cup L$ we have that $\mu \models R$ iff $G\models_l ft(R)\mu$,  $G\models_l m(\textsc{opt}(P_1, P_2, R))\mu$ iff $\mu\in  \llbracket \textsc{opt}(P_1, P_2, R) \rrbracket_G$, and  $G\models_l m(\textsc{union}(P_1, P_2))\mu$ iff $\mu\in  \llbracket \textsc{union}(P_1, P_2) \rrbracket_G$, 

 %$G\models_l (\{\langle m(P_1), \text{opt}, m(P_2)\rangle\}\cup ft(R) )\mu$ iff $\mu\in  \llbracket \textsc{OPT}(P_1, P_2, R) \rrbracket_G$ and $G\models_l \{\langle m(P_1), \text{union}, m(P_2)\rangle\} )\mu$ iff $\mu\in  \llbracket \textsc{UNION}(P_1, P_2) \rrbracket_G$
%To deal with the local scoping in the arguments of the  predicates $\includes$ and $\notIncludes$ we furthermore assume that the mapping $m$ additionally replaces all locally scoped universal variables by blank nodes. 
Under the above assumption we show:

\begin{lemma}\label{lemma:1}
	Given a proper RDF graph   $G$, a SPARQL graph pattern $P$, and a mapping $\mu: sv(P) \rightarrow I \cup L$, then the following holds:
	%There exists an extension $\mu^\prime: var(P) \rightarrow I_G\cup B_G \cup L_G$ such that $\mu^\prime|_{sv(P)}=\mu$ and
%then following holds:
	\[\mu\in \llbracket P \rrbracket_G \text{  iff  }  G\models_l m(P)\mu.\]
\end{lemma}

%We provide a proof for the lemma in our online appendix \cite{extended}.
\begin{proof}
We first note that the query patterns form \textsc{minus}, \textsc{fne} and \textsc{fe} rely on the predicates $\includes$ and $\notIncludes$ which are subject to local variable  scoping. These are exactly the constructs considered in this proof which take multiple patterns as arguments and only return  mappings their first argument. Therefore there are no free variables in the terms $m(P)\mu$ if $\mu$ is defined on the scope of $P$. Below, we therefore do not have to consider the case that $m(P)\mu$ contains free variables.
We show the claim by induction over the structure of $P$:
\begin{itemize}
\item	Let $P$ be a BGP:
\begin{description}
	\item["$\Rightarrow$"] Let $\mu\in \llbracket P \rrbracket_G$  and let $\mathfrak{I}\models_l G$.  As $P$ is a BGP, we have  $m(P)=P$ and $P\mu\subseteq G$, thus $\mathfrak{I}\models_l  m(P)\mu$. 
\item["$\Leftarrow$"] Let $G\models_l m(P)\mu$.  As $G$ and $P\mu$ are proper RDF graphs, $P\mu$ must be an instance of $G$ (for plain RDF graphs N3 logical entailment behaves like RDF's simple entailment).  We thus get that $P\mu\subset G$ and thereby
$\mu \in \llbracket P \rrbracket_G$.
	\end{description}
\item Let $P$ be of the form $\textsc{AND}(P_1,P_2)$:
\begin{description}
\item["$\Rightarrow$"]	
	Let  $\mu\in \llbracket P \rrbracket_G$  and let $\mathfrak{I}\models_l G$. 
	We have that $m(P)=P_1\cup P_2$ and $\llbracket P \rrbracket_G = \llbracket P_1 \rrbracket_G \bowtie \llbracket P_2 \rrbracket_G$. As $\mu|_{sv(P_1)}\in \llbracket P_1 \rrbracket_G$ and $\mu|_{sv(P_2)}\in \llbracket P_2 \rrbracket_G$ we get by the induction hypothesis that 
	$\mathfrak{I}\models_l m(P_1)\mu$ and 	$\mathfrak{I}\models_l m(P_2)\mu$ and thus $\mathfrak{I}\models_l m(P)\mu$.
\item["$\Leftarrow$"]	Let $G\models_l m(P)\mu$, and let $\mathfrak{I}\models_l G$.  As $m(P)=m(P_1)\cup m(P_2)$, 
we can define $\mu_1:= \mu|_{sv(P_1)}$ and $\mu_2:= \mu|_{sv(P_2)}$.  These two mappings are compatible and because of
$\mathfrak{I}\models_l m(P_1)\mu_1$ and $\mathfrak{I}\models_l m(P_2)\mu_2$, we get by the induction hypothesis that
 $\mu_1\in \llbracket P_1 \rrbracket_G$ and $\mu_2\in \llbracket P_2 \rrbracket_G$.
\end{description}	
\item	Let $P$ be of the form $\textsc{MINUS}(P_1,P_2)$:

	%	and
	%
	%	$\mu\in  \llbracket P_1 \rrbracket_G\setminus_M  \llbracket P_2 \rrbracket_G = \{\mu \in \ \llbracket P_1 \rrbracket_G \mid \text{for all } \mu^\prime \in  \llbracket P_2 \rrbracket_G: \mu \not\sim\mu^\prime\text{ or } \text{dom}(\mu)\cap \text{dom}(\mu^\prime)=\emptyset  \}$.  We consider the cases that $\mu\not\sim\mu\prime$ and $\text{dom}(\mu)\cap \text{dom}(\mu^\prime)=\emptyset$ separately:
	%
 
	\begin{description}
			\item["$\Rightarrow$"]
			Let  $\mu\in \llbracket P \rrbracket_G$  and let $\mathfrak{I}\models_l G$. 
			and $\mu\in  \llbracket P_1 \rrbracket_G\setminus_M  \llbracket P_2 \rrbracket_G $. The latter means that
			%If $\mu \in  \llbracket P_1 \rrbracket_G\setminus_M  \llbracket P_2 \rrbracket_G $, then  
			$\mu \in  \llbracket P_1 \rrbracket_G$  and for all solution mappings $\mu^\prime\in \llbracket P_2 \rrbracket_G$ we either have that $\text{dom}(\mu)\cap \text{dom}(\mu^\prime)=\emptyset$ or $\mu \not\sim\mu^\prime$. 
			We consider these two cases separately
			
	\begin{itemize}
		\item If $\text{dom}(\mu)\cap \text{dom}(\mu^\prime)=\emptyset$  for all solution mappings  $\mu^\prime \in  \llbracket P_2 \rrbracket_G$, then $sv(P_1)\cap sv(P_2)=\emptyset$, and we have that $m(P)=m(P_1)$. The claim follows by the induction hypothesis as $\mathfrak{I}\models_l m(P_1)\mu$ for $\mu \in \llbracket P_1 \rrbracket_G$.
		\item  Let  $\mu \not\sim\mu^\prime$ for all solution mappings
		  $\mu^\prime \in  \llbracket P_2 \rrbracket_G$. To show that $\mathfrak{I}\models_l  G\cup (m(P_1)\cup \{\langle \_, \notIncludes, rl(m(P_2),sv(P_2))\rangle\})\mu$ it suffices to show that $\mathfrak{I}\models_l m(P_1)\mu$  and $\mathfrak{I}\models_l  \{\langle \_, \notIncludes, rl(m(P_2),sv(P_2))\rangle\}\mu$. We can treat these two parts separately 
		  as the triple $\langle \_, \notIncludes, rl(m(P_2),sv(P_2))\rangle$ has graph terms as arguments and all blank nodes and variables occurring in these are scoped locally within the graph terms. We know that $\mathfrak{I}\models_l m(P_1)\mu$ by the induction hypothesis. It remains to show that $\mathfrak{I}\models_l  \{\langle \_, \notIncludes, rl(m(P_2),sv(P_2))\rangle\}\mu$.  We do that by contradiction assuming that $\mathfrak{I}\models_l \{\langle \_, \includes, rl(m(P_2),sv(P_2))\rangle\}\mu$.  Let $blank(P_2)$ be the set of all blank nodes occurring in $P_2$.  Then there exists a mapping $\sigma: var(m(P_2))\cup blank(m(P_2))\rightarrow I_G\cup B_G\cup L_G$ such that  $G\models_l m(P_2)(\mu \cup \sigma)$  with $\mu^\prime=(\mu\cup\sigma)|_{sv(m(P_2))}$ and the induction hypothesis we get a mapping $\mu^\prime \in  \llbracket P_2 \rrbracket_G$ which is compatible with $\mu$. This contradicts the initial assumption.
	\end{itemize}
	\item["$\Leftarrow$"] 	Let $G\models_l m(P)\mu$, $\mathfrak{I}\models_l G$ and $\mu: sv(P)\rightarrow I\cup L$. % and let $\mathfrak{I}\models_l G$. 
	We note that in this case $sv(P)=sv(P_1)$ and
	\[m(P)=\begin{cases}
		m(P_1), \text{ if } sv(P_2)\cap sv(P_1)= \emptyset\\
		m(P_1)\cup \{\langle \_, \notIncludes, rl(m(P_2),sv(P_2))\rangle 
		\}, \text{ else},
	\end{cases}
	%\cup \{P_1|P_1\in m(Q_1) \wedge   P_2\in m(Q_2)\wedge
	% sv(P_2)\cap sv(P_1)= \emptyset\}
	\] 
	We again consider the two cases separately:
	\begin{itemize}
		\item If $sv(P_1)\cap sv(P_2)=\emptyset$ then $\text{dom}(\mu)\cap \text{dom}(\mu^\prime)=\emptyset$  for all solution mappings  $\mu^\prime \in  \llbracket P_2 \rrbracket_G$. We thus have to show that $\mu \in  \llbracket P_1 \rrbracket_G$. This follows directly for the induction hypothesis as
	$m(P)=m(P_1)$ and  $G\models_l m(P_1)\mu$. 
	%	 We get by the induction hypothesis that $\mu \in  \llbracket P_1 \rrbracket_G$
		\item If $sv(P_1)\cap sv(P_2)\neq\emptyset$,  
		then
		$m(P)=m(P_1)\cup \{\langle \_, \notIncludes, rl(m(P_2),sv(P_2))\rangle\}$. We have to show that $\mu \in  \llbracket P_1 \rrbracket_G\setminus_M  \llbracket P_2 \rrbracket_G$ which in this case means that for all solution mappings  
		$\mu^\prime \in \llbracket P_2 \rrbracket_G$ we have that $\mu\not\sim \mu^\prime$. We do that by contradiction and assume that there exists a mapping  $\mu^\prime \in  \llbracket P_2 \rrbracket_G$ with $\mu\sim \mu^\prime$.  By the induction hypothesis, this  means that $G\models_l m(P_2)\mu^\prime$ and as $\mu\sim \mu'$ also that $G\models_l (m(P_2)\mu)\mu''$ where $\mu'':dom(\mu')\setminus dom(\mu)\rightarrow I\cup L$ and is defined by $\mu''(x):=\mu'(x)$ for all $x\in dom(\mu'')$.  With that we get that $G\models_l (m(P_2)\mu)\mu^{\prime\prime}$.  Following the definition of $\includes$, this means that 
	   $\mathfrak{I}\models  \{\langle close(G) \includes, m(P_2))\rangle\}\mu$ but as $rl(m(P_2),sv(P_2)$ only differs from 
	   from $m(P_2)$ by the variable which are in the domain of $\mu''$ we define a mapping $rn$ from the relabeled  variables occurring in $rl(m(P_2), sv(P_2))$ to the variables occurring in $m(P_2)$, with $\mu''\circ rn$, we get that 
	   $\mathfrak{I}\models  \{\langle close(G) \includes, rl(m(P_2),sv(P_2)))\rangle\}\mu$  which in our context means that $\mathfrak{I}\models  \{\langle \_ \includes, rl(m(P_2),sv(P_2)))\rangle\}\mu$ which contradicts that $\mathfrak{I}\models \{\langle \_, \notIncludes, rl(m(P_2),sv(P_2))\rangle\}\mu $. 
	\end{itemize}	
\end{description}	
	%	
	%	We go trough the definition of $\mu\in  \llbracket P_1 \rrbracket_G\setminus_M  \llbracket P_2 \rrbracket_G $ and consider the two cases that for all solution mappin  $\text{dom}(\mu)\cap \text{dom}(\mu^\prime)=\emptyset$, or that all solution mappings $\mu\prime \in $
	%We consider the two cases separately: If $sv(Q_1)\cap sv(Q_2)=\emptyset$
\item For
 $P$ of the form $\textsc{FNE}(P_1,P_2)$, the argumentation is analogous to the second case of the minus construct. 
 %The main argument is that if $\llbracket P_2\mu \rrbracket_G=\emptyset$, then there is also no
%\begin{description}
 %	\item["$\Rightarrow$"] 	Let  $\mu\in \llbracket P \rrbracket_G$  and let $\mathfrak{I}\models_l G$. Then $ \mu\in \llbracket P_1 \rrbracket_G$ and $\llbracket P_2\mu \rrbracket_G=\emptyset$.  By induction hypothesis $\mathfrak{I}\models m(P_1)$
 \item For
 $P$ of the form $\textsc{FE}(P_1,P_2)$ we can use the same argumentation as for \textsc{and}. The only difference is that $\includes$ imposes local scoping. 
 %	We assume that $\mathfrak{I}\models P_2\mu$, then 
 %	\item["$\Leftarrow$"] If 
% \end{description}
 %then $\mu\in  \llbracket P_1 \rrbracket_G$, $\llbracket P_2\mu \rrbracket_G=\emptyset$ and
%	$m(P)=m(P_1) \cup \{\langle \_, \notIncludes, m(P_2)\rangle\}$. Using the same explanation as for the second case of the previous construct, $\langle \_, \notIncludes, m(P_2)\rangle$ is equivalent to 
%	$\langle G, \notIncludes, m(P_2)\rangle$ where the variables $v\in sv(P_2)\setminus sv(P_1)$ are scoped locally and treated as blank nodes.  The triple $\langle \_, \notIncludes, m(P_2)\rangle$ occurring in the premise of the rule is thus true if $\llbracket P_2\mu \rrbracket_G=\emptyset$.  As above, the claim follows.
	%	If $P$ is of the form $\textsc{FE}(P_1,P_2)$ then $\mu\in  \llbracket P_1 \rrbracket_G$, $\llbracket P_2\mu \rrbracket_G\neq\emptyset$ and
%	$m(P)=m(P_1)\cup m(P_2)$.  As $\llbracket P_2\mu \rrbracket_G\neq\emptyset$, there exists a mapping $\mu^2\in  \llbracket P_2\mu \rrbracket_G$ and  we get $\mu^\prime=\mu\cup \mu^2$ and 
%	$\mathfrak{I}\models m(P)\mu^\prime$.
	%	If $P$ is of the form $\textsc{FILTER}(P_1,R)$ then$\mu \in \llbracket P_1 \rrbracket_G$, $\mu \vDash R$ and
	%$m(Q)=m(Q_1)\cup ft(R)$.
	%If we assume that the translation function works correctly, that is that $\mu \vDash R$ iff $\mathfrak{I}\models ft(P)\mu$ then the claim  follows immediately. \todo{Think about that one!}
\item	The claim for the remaining three constructs $\textsc{filter}(P_1, P_2, R)$, $\textsc{UNION}(P_1,P_2)$ and $\textsc{OPT}(P_1,P_2,R)$ rely on the run-time functions which we assume for this proof to work correctly.  \qed
	%	If $P$ is of the form $\textsc{UNION}(P_1,P_2)$
%	\todo{maybe add condition that filter, optional and union work as expected}
	%		If $P$ is of the form $\textsc{OPT}(P_1,P_2,R)$
%	 By the induction hypothesis we know that 
	%	$\mathfrak{I}\models_l m(P_1)\mu_1$ and $\mathfrak{I}\models_l m(P_2)\mu_2$ for all $\mu_1\in\llbracket P_1 \rrbracket_G$ and $\mu_2\in\llbracket P_2 \rrbracket_G$ and as
	\end{itemize}
%	\end{description}
	\end{proof}

\subsection{From queries to consequences}\label{sec:consequence}
%To construct the conclusion of the N3 rule, we define a our translation function our scoping function $sv$ 
Given a SPARQL query $(n~x) \text{ WHERE } P$, we define the function $h$ which takes 
the query form $f=(n~x)$ and the SPARQL graph pattern $P$ as input and produces an N3 graph pattern as output.
%, that is, we see query forms as a pair $(n, x)$, where $n \in \{\textsc{select},\textsc{construct}\}$ is the name of the query form, $x$ is its argument, that is, it is either a list of variables or a star if $n=\textsc{select}$, or a BGP, if $n=\textsc{construct}$, and $Q$ is a SPARQL graph pattern. 

Let $F$ be the set of query forms and ``$\lis$'' a function which, provided with a set of variables $V$, produces an ordered list of these.  We define the head function $h: F\times GP\rightarrow N3P$ as follows:
\[
h(f,P)=\begin{cases}
Q, \text{ if } f=(\textsc{construct}, Q),\\
\{\langle \_,\result,l\rangle\}, \text{ if  } f=(\textsc{select}, l), \text{ if } l\neq *,\\
\{\langle \_,\result,\lis(sv(P))\rangle\}, \text{ if  } f=(\textsc{select}, *).
\end{cases}
\]
where $\_$ stands for an arbitrary blank node. %This use of an arbitrary blank node helps us to simulate the multiset semantics as for each solution mapping a new 
Note that we use the scoping function $sv$ from above to resolve the asterisk.

For our example query $P$ with its concrete realisation in Listing~\ref{lst:minus-example}, our mapping produces the consequence\footnote{In our implementation, the result triple also includes the original \textsc{select} variable names as strings, so query results are returned in a way similar to SPARQL.} 
%$P=\text{MINUS}(\{\langle?x, p, ?n\rangle\}, \text{FE}(\{\langle ?x,q,?m\},\{?m, r, ?n\}))$
\begin{equation}
h((SELECT, *), P)=\{\langle \_, \result, \lis(sv(P))\rangle=\{\langle \_, \result, (?x\ ?n)\rangle\}.
\end{equation}
Combining the results from applying the functions $m$ and $n$ to the query in Listing~\ref{lst:minus-example} we get the following N3 rule as a translation:
\[
\{ \langle ?x, p, ?n \rangle, \langle \_ ,\notIncludes ,\{\langle ?x, q, ?m\rangle, \langle ?m, r, ?n_{new}\rangle \}\rangle\} \rightarrow \langle \_, \result, (?x\ ?n)\rangle
\]

We can use Lemma \ref{lemma:1} to prove the following:
\begin{lemma}
	Given a query $(\textsc{Construct} ~ Q) \textsc{where}~P$ and  a proper RDF graph $G$ : If  $\mu \in \llbracket P \rrbracket_G$, then $G\cup \{m(P)\rightarrow h((\textsc{Construct}~Q), P) \}\models_l Q\mu $
\end{lemma}

%We provide a proof for the lemma in our online appendix \cite{extended}.
\begin{proof}	We first note, that $h((\textsc{Construct}~Q), P)=Q$.  
	Let $\mathfrak{I}\models_l G\cup \{m(P)\rightarrow Q\}$.
	From $\mu \in \llbracket P \rrbracket_G$
follows by  Lemma~\ref{lemma:1} that  $\mathfrak{I}\models_l m(P)\mu$.
As $\mathfrak{I}\models m(P)\rightarrow Q$ and the variables in the rule are universally quantified, we can apply the rule using a mapping $A^\mu$ depending on $\mu$. %If $\mu$ maps a variable to a blank node, we make use of the fact that with the fragment of SPARQL considered in this paper all variables are bound to nodes occurring in $G$ and that  $\mathfrak{I}\models G$.  The latter means that for the blank nodes in $G$ there exists a mapping $A^G$ such that $\mathfrak{I}[A^G]\models_l G $. 
We define $A^\mu(x):=(\mathfrak{I}(\mu'(x)) ,\mu'(x))$.
If $Q\mu$ contains blank nodes, then these are quantified locally and they do not depend on the premise. \qed
\end{proof}

With the lemma we can apply a single translated rule to a ground RDF graph and obtain the same triples the original query provides. 
When applying the rules recursively, we rely on N3 semantics.  Note that the result can easily be extended to RDF graphs containing blank nodes as N3 rules never change the scoping of blank nodes in plain RDF graphs and these are quantified on graph level.
%The result can be easily extended to RDF graphs containing blank nodes as blank nodes in the data graph are quantified on this graph and therefore act like constants in the rules. The restriction was only made to simplify the proofs.

\subsection{SPARQL Concepts in N3}
\label{sec:n3_special}

Our goal is to map one SPARQL query to one N3 rule, however, some SPARQL concepts do not have an equivalent N3 construct. 
To address this mismatch, we introduce a dedicated set of auxiliary \emph{runtime rules}, expressed in native N3, which must be evaluated alongside the translated rules to achieve the intended results. 
% We implement those as native N3 constructs by providing an extra set of \emph{runtime-rules}, which 
% As mentioned, our approach handles some SPARQL constructs, which disallow mapping to a single N3 rule, by a special set of rules called \emph{runtime-rules}. 
% need to be applied together with the translated rules to obtain the desired result. 
These runtime rules can also be integrated with native N3 rules, providing direct support for SPARQL-specific features such as \textsc{union}, \textsc{optional}, and property paths within N3; thus providing a true cross-fertilization.  We explain the rules for \textsc{union} and  \textsc{optional} below.%
\footnote{
	To make the assumption in Section \ref{sec:premise} true, we need to provide a proof that the rules work exactly like their corresponding SPARQL pattern.  We  only provide an explanation here to at least make the rules plausible to the reader.  The main reason for omitting the proof is that the rules for \textsc{optional} are defined using the predicate $\cop$. This predicate has no model theoretic semantics yet. We plan to provide  a formalisation and a proof in future work.
}

% We note that they can also be used directly with ``vanilla'' (i.e., non-translated) N3 rules, making it possible to use the SPARQL concepts \textsc{union}, \textsc{optional} and property paths directly in N3. 
% We discuss these three aspects separately starting with \textsc{union}. 

For \textsc{union}, a single translated N3 rule will include a triple {\small \texttt{<a> union <b>}}. This triple will be resolved by the following runtime rules:
\begin{align}
\langle x,\unio, y\rangle \leftarrow \langle \_, \includes, x \rangle 
\hspace{2cm} \langle x,\unio, y\rangle \leftarrow \langle \_, \includes, y \rangle
\end{align}\

% For \textsc{union}, the runtime includes the rules:
% \begin{align}
% \langle x,\unio, y\rangle \leftarrow \langle \_, \includes, x \rangle \\
% \langle x,\unio, y\rangle \leftarrow \langle \_, \includes, y \rangle
% \end{align}
% Here, a single translated N3 rule will include a triple {\small \texttt{<a> union <b>}}, which will be resolved by these runtime rules.
% The evaluation of these rules is initiated 
These rules are evaluated using the subject {\small \texttt{<a>}} and object {\small \texttt{<b>}} of the union triple, and succeeds if at least one of the corresponding subgoals is satisfied.
% The backward rules are called with subject {\small \texttt{<a>}} and object {\small \texttt{<b>}} of the union, succeeding if at least one of the two calls succeeds. 
By that, we mimic disjunction. 
Note that we also support property paths but consider them out of scope for this paper; we address them in future work.
%

%The rules for the 
%\da{Plan: optional rules go in. But maybe simplified}
\begin{lstlisting}[caption={Example query demonstrating the behavior of \textsc{optional}.}, label={lst:opt-example}, float=t]
SELECT * { 
  :x1 :p ?v . OPTIONAL { :x2 :q ?w .
    OPTIONAL { :x3 :p ?v }}}
\end{lstlisting}
The rules for \textsc{optional} are more complicated. We show them in Figure~\ref{fig:optional}.
 % We display them in Figure~\ref{fig:optional}, but to save space, we only discuss them briefly. 
 % (wvw) trying to tie this into the problem (not solution) with MINUS, which seems very similar
As with the \textsc{minus}, we illustrate the complexity of the \textsc{optional} applying an example query (Listing~\ref{lst:opt-example}) to a data graph $G= \{\langle x1, p, 1\rangle,\langle x2, q, 2  \rangle, \langle x3, p, 3 \rangle\}$.

Here, the variable \texttt{?v} will be bound twice with different values; once with 3 in the most deeply nested \textsc{optional} clause, and once with 1 in the mandatory clause. 
For that reason, SPARQL considers the variable bindings of the mandatory and nested \textsc{optional} as non-compatible, and only the former are returned.
Similar as with the \textsc{minus}, however, N3 does not allow different bindings for the same variable in the same premise. 
We again solve the problem by relabelling variables: we represent each $X \in N3GP$  as a pair $(X_P, X_I)$, where $X_P$ is the pattern with the original variables, and $X_I$ is a copy of $X_P$ with ``fresh'' variables to be instantiated. We then use the latter to find matches in the data without variable conflicts. We use the N3 ``copy'' predicate to create $X_I$ from $X_P$.

The \textsc{optional} in SPARQL is implemented using a nested left join.
On the left-hand side, the first rule first invokes the predicate ``$\eval$'', which retrieves the final pattern ($X_P$) and instantiation ($X_I$) of the solution. Here, the next triple with predicate ``$\includes$'' will instantiate the variables in $X_P$ using the values from $X_I$\footnote{Note that this predicate thus also binds variables in the argument graphs.}.
The last two rules with the ``$\unnest$'' predicate deal with nested \textsc{optional} clauses, and thus implement the recursive nature of the left join.
On the right-hand side, the rules implement the actual left join. 
Before binding variables, they use the ``copy'' N3 builtin to create copies $X_I$ of patterns as described above. They further rely on the ``incl'' and ``notIncl'' predicates as described earlier.
 Here, we consider three cases: (1) Solution mappings for both \textsc{optional} arguments are compatible and fulfill the filter condition; (2) There is a solution mapping for the first argument, but not for the second; or, the union of both mappings does not fulfill the filter condition; (3) There is no solution mapping for the first argument, in which case an empty solution pattern is retrieved.

\begin{figure}[t] 
\noindent
\scriptsize
\begin{minipage}[t]{0.45\textwidth}
\begin{align*}
  \langle M, \opt,&O \rangle \\
\leftarrow & \langle (M, O), \eval, (X_P,X_I) \rangle, 
  \langle X_P, \includes, X_I \rangle\\[5pt]
  \langle (M, O),& \eval, (Z_P, Z_I) \rangle\\ 
  \leftarrow & \langle M, \unnest, (X_P, X_I)\rangle, \\
 & \langle O, \unnest, (Y_P, Y_I) \rangle,\\
  &\langle ((X_P, X_I), (Y_P, Y_I)) \leftjoin, (Z_P, Z_I) \rangle\\[5pt]
  \langle X, \unnest,& (Y_P, Y_I) \rangle \\
  \leftarrow & \langle X, \equalTo, \langle M \opt O \rangle \rangle,\\
  &\langle (M,O), \eval, (Y_P, Y_I)\rangle\\[5pt]
  \langle X, \unnest,&(X,\{\}) \rangle \ \\
  \leftarrow & \langle X, \notEqualTo, \langle M \opt O \rangle \rangle
\end{align*}
\end{minipage}
\hspace{0.5cm} 
\hfill
\begin{minipage}[t]{0.45\textwidth}
\begin{align*}
  \langle ((M_P,& M_I),(O_P,O_I)), \leftjoin, (R_P, R_I) \rangle\\ 
  \leftarrow & \langle (M_P, O_P), \conj, R_P \rangle, \\
 &\langle R_P, \cop, R_I \rangle, \langle \_, \includes, R_I \rangle, \\ 
  &\langle R_I, \includes, M_I \rangle,
  \langle R_I, \includes, O_I \rangle\\[5pt]
  \langle ((M_P,& M_I),(O_P,O_I)), \leftjoin, (M_P, R_I) \rangle\\
  \leftarrow & \langle M_P, \cop, R_I\rangle, 
  \langle \_, \includes, R_I\rangle,\\
  &\langle (M_P, O_P), \conj, C_P\rangle, 
  \langle C_P, \cop, C_I\rangle,\\
  & \langle C_I, \includes, R_I\rangle, 
  \langle \_, \notIncludes, C_I\rangle\\[5pt]
  \langle ((M_P,& M_I),(O_P,O_I)),\leftjoin, (\{\}, \{\}) \rangle\\
  \leftarrow & \langle \_, \notIncludes, M_P\rangle
\end{align*}
\end{minipage}
  \caption{\normalsize Runtime rules for \textsc{optional}.}
  \label{fig:optional}
\end{figure}

%As one of the goals of this implementation is to truly connect N3 and SPARQL, we defined rules to handle special concepts of SPARQL, like the optional operator or the property paths. Such concepts can also be handled already at compile time, that is, when the translation of SPARQL queries is produced. However, we hope that this work will motivate SPARQL users to write their rules directly in N3. This can be stimulated by directly supporting SPARQL concepts.
%\dt{ Namespace for SPARQL concepts in N3 \url{https://notation3.org/sin3\#}} 
%\input{extensions}
% \input{implementation_iswc}
\section{Evaluation}
\label{sec:eval}
%\input{implementation} 
% \da{maybe we should rephrase, because we do not evaluate our system but the translation}
% \da{What is the goal of our evaluation? I would put that here. Maybe whether N3 systems can compete with state of the art? Maybe test "the potential" of the approach? }
We implemented the above SPARQL-in-N3 (SiN3) translation and compare its performance with state-of-the-art systems. All code is available on github~\cite{onlineappendix}, including an online demo of the Zika screening use case (see below).
%\da{say why} % and checking its compliance with SPARQL 1.1.
% We further show the impact of different reasoning strategies, afforded by our use of a rule language, on performance.

%\da{buzz word, somehow the reader decides whether it is comprehensive} 
% benchmarking its recursive reasoning capabilities against state-of-the-art approaches using real-world RDF datasets and a healthcare use case, with a focus on performance, reasoning strategies, and SPARQL 1.1 compliance.

\subsection{Evaluation setup}
\label{sec:eval_setup}

\vspace{0.2cm}
\noindent \textbf{Datasets and queries}\\
Our evaluation covers the following use cases:

\vspace{0.1cm}
\textit{LMDB and YAGO}. We apply a similar setup as recSPARQL~\cite{reutter2021}, which used the \textit{LMDB} (Linked Movie Database), an RDF dataset about movies and actors~\cite{hassanzadeh2009linked} (6,148,121 triples); and \textit{YAGO} (Yet Another Great Ontology), RDF data including people, locations, and movies~\cite{pellissier2020yago} (3,000,006 triples).
Movie-related recursive \textsc{construct} queries were formulated on both datasets:

% \begin{description}
    \textit{Bac1}. Searches for actors with a finite ``Bacon number'', i.e., they co-starred in the same movie with Kevin Bacon or another actor with such a number. \\
\indent \textit{Bac2}. Searches for actors with a finite Bacon number whereby all collaborations were done in movies with the same director. \\
\indent \textit{Bac3}. Searches for actors connected to Kevin Bacon through movies where the director is also an actor (not necessarily in the same movie).
% \end{description}

\vspace{0.1cm}
\noindent They further formulated the following queries on the YAGO dataset:\\
% \begin{description}
\indent \textit{Geo}. Searches for places in which the city of Berlin is (transitively) located.
\indent \textit{MarrUs}. Searches for people who are transitively related to someone through the \textit{isMarriedTo} relation, who owns property within the United States.
% \end{description} \\
Hence, these queries mostly calculate the transitive closure of a given relation.

% \vspace{0.05cm}
% We note that recSPARQL was also evaluated using the \textit{gmark} benchmark~\cite{gmark_2017}. As the efficient evaluation of property paths is still a work in progress, it is considered out of scope for this paper. We revisit this issue in future work.

\vspace{0.05cm}
\textit{Zika screening}. We base ourselves on a prior healthcare use case on Zika screening~\cite{sin3_demo}. 
Using \textsc{construct} queries, we implemented the CDC testing guidance for Zika, 
% \footnote{https://www.cdc.gov/zika/hc-providers/testing-guidance.html}
which determine whether a patient should be tested for Zika, based on a series of factors. %such as pregnancy, symptoms, recent travel to Zika areas, or sexual contacts with residents or travelers to Zika areas. 
% This use case illustrates the benefits of general recursion in SPARQL, namely modular queries; 
% We represent patient data using the HL7 FHIR~\cite{HL7FHIR} vocabulary, an EHR interoperability standard that offers an RDF representation (Turtle syntax).
We expand on this test case with randomly generated datasets, different versions of the used HL7 FHIR vocabulary~\cite{HL7FHIR}, and reasoning over a biomedical ontology, i.e., SNOMED~\cite{SNOMED}.
% This use case illustrates how general recursion in SPARQL, as enabled by our translation, allows for modular queries; each Zika indicator is checked by an individual query, and its results are in some cases re-used by other queries. 
% whereby the results of multiple queries, such as checking for travellers to Zika areas, are re-used by other queries. 
% In a similar vein, we provide recursive utility queries to support more easy navigation of FHIR.

\vspace{0.05cm}
\textit{Availability}. All datasets and queries are in our git repository \cite{onlineappendix}.

\vspace{0.2cm}
\noindent \textbf{Evaluated systems}\\
We evaluate our SPARQL-in-N3 translation, labelled \textit{sin3}, using the EYE v10.24.10 reasoner~\cite{eyepaper}, which can perform both forward and backward reasoning.
% (wvw) one phrase summary!
Our prior demo paper~\cite{sin3_demo} %\cite{sin3_demo} 
outlines the implementation of \textit{sin3}, including the different execution steps and rulesets involved.
We compare with two other systems that utilize \textsc{construct} queries for rule-based reasoning, namely \textit{spinrdf}~\cite{spinrdf} and \textit{recSPARQL}~\cite{reutter2021}.
The work by Polleres~\cite{polleres2007} was implemented using an older version of dlvhex %({\footnotesize http://www.dlvsystem.com/}) 
which we found to no longer be compatible with modern OS. 
As far as we know, the work by Gottlob et al.~\cite{gottlob2015beyond} was not implemented; and, as mentioned, the SparqLog~\cite{angles2023} does not support recursive querying.

\vspace{0.1cm}
\noindent \textbf{Hardware}
The experiments were conducted on a MacBook Pro with an Apple M1 Pro processor, 32 GB of RAM, and a 1 TB SSD, running macOS Sonoma 14.6.1. Each experiment was executed 5 times and results were averaged. %to ensure reliable results, and the performance metrics were averaged over these runs.

\vspace{0.1cm}
\noindent \textbf{Compliance}
Regarding compliance with the SPARQL 1.1 specification~\cite{harris2013sparql}, we currently lack support for the \textsc{having} clause, named graphs, expressions in \textsc{bind} and \textsc{select} clauses, and \textsc{limit} clauses in subqueries. We currently support only the \textsc{construct} and \textsc{select} query forms.

\subsection{Results}

\begin{table}[t]
\label{tbl:lmdb_results}
\centering
\begin{tabular}{|l?c|l|l|l?c|l|l?c|l|}
    \hline
    \multicolumn{10}{|c|}{LMDB} \\ \hline
    \multirow{2}{*}{\textbf{Query}} &
      \multicolumn{4}{c?}{\textbf{sin3}} &
      \multicolumn{3}{c?}{\textbf{spinrdf}} &
      \multicolumn{2}{c|}{\textbf{recSPARQL}} \\ \cline{2-10}
    & Load & Gen SPIN & Gen N3 & Exec & Load & Gen SPIN & Exec & Load (TDB) & Exec \\ \hlineB{2}
Bac1 & \multirow{3}{*}{31} & 0.58 & 0.10 & 20.1 & \multirow{3}{*}{26.9} & 0.06 & 12   & \multirow{3}{*}{75.4} & 65.7 \\ \cline{1-1}\cline{3-5}\cline{7-8}\cline{10-10}
Bac2 &                       & 0.58 & 0.08 & 2.4  &                       & 0.02 & 3.9  &                       & 18.4 \\ \cline{1-1}\cline{3-5}\cline{7-8}\cline{10-10}
Bac3 &                       & 0.58 & 0.08 & 9.0  &                       & 0.02 & 19.5 &                       & 24.5 \\ \hlineB{2}

    \multicolumn{10}{|c|}{YAGO} \\ \hline
Bac1   & \multirow{5}{*}{36.1}   & 0.58 & 0.08 & 14.3 & \multirow{5}{*}{44.9}   & 0.02  & 7.5   & \multirow{5}{*}{113686} & 50.2 \\ \cline{1-1}\cline{3-5}\cline{7-8}\cline{10-10}
Bac2   &                         & 0.57 & 0.08 & 1.5  &                         & 0.003 & 0.01  &                         & 3.2  \\ \cline{1-1}\cline{3-5}\cline{7-8}\cline{10-10}
Bac3   &                         & 0.57 & 0.08 & 5.0  &                         & 0.003 & 2.8   &                         & 22.4 \\ \cline{1-1}\cline{3-5}\cline{7-8}\cline{10-10}
Geo    &                         & 0.57 & 0.07 & 0.4  &                         & 0.003 & 0.004 &                         & 14.4 \\ \cline{1-1}\cline{3-5}\cline{7-8}\cline{10-10}
MarrUs &                         & 0.57 & 0.08 & 12.3 &                         & 0.003 & 5.1   &                         & 83.2 \\ \hline
\end{tabular}

\caption{Execution times for LMDB and YAGO queries (times in seconds).}
\end{table}

We only show and discuss the results for the \textit{LMDB} and \textit{YAGO} experiments for space considerations. 
Our git repository~\cite{onlineappendix} details the \textit{Zika} experiment.

Both \textit{sin3} and \textit{spinrdf} involve converting input SPARQL queries into SPIN code (Gen SPIN). 
Subsequently, \textit{sin3} converts the SPIN code into N3 rules (Gen N3). 
% The results for the LMDB and YAGO queries are shown in Table~\ref{tbl:lmdb_results} and Table~\ref{tbl:yago_results}, respectively.
We observe that Gen SPIN takes much longer for \textit{sin3} than \textit{spinrdf}, as the former requires starting a separate JVM for generating SPIN code. %, whereas the latter is entirely executed in Java.
\textit{Sin3} and \textit{spinrdf} load the dataset into memory (Load), 
whereas \textit{recSPARQL} pre-creates a persistent Jena Triple DataBase (Load TDB). 
% Such persistent storage will work better for very large datasets; %, however, both the LMDB and Yago datasets fit into memory. 
% however, reasoning performance (Exec) will be better with in-memory data. %compared to a persistent database.
% While these are relatively large datasets (LMDB is ca. 892Mb, Yago is 982Mb), 
\textit{Sin3} is competitive regarding reasoning performance (Exec): for the LMDB dataset, on average, \textit{sin3} takes ca. 10.5s, \textit{spinrdf} ca. 11.8s, and recSPARQL ca. 36.2s. 
For the Yago dataset, on average, \textit{sin3} takes ca. 6.7s , \textit{spinrdf} ca. 3.1s, and recSPARQL ca. 34.7s. 
Hence, we note that \textit{spinrdf}, which uses a query engine (Jena), performs much better for the Yago dataset. We revisit this in future work. 

A benefit of rule languages is that both forward- and backward reasoning strategies can be used.
To illustrate this, we used the Deep Taxonomy (DT)\footnote{\url{https://eulersharp.sourceforge.net/2009/12dtb/}}, a \textsc{construct} query that implements the OWL2 RL \textit{cax-sco} rule~\cite{OWL2RL}, and a custom query requesting all instances of class \texttt{A2}. Backward reasoning has an advantage here, as only a small part of the transitive closure has to be searched; forward reasoning has to materialize the entire transitive closure. Indeed, after translating the OWL2 RL query to a backward-chaining N3 rule, \textit{eye} takes avg. 34ms, while \textit{spinrdf} does not return results after 1h (results not shown in table). 

\section{Conclusions}
\label{sec:conclusions}

% (wvw) note to self
% - conseq of single mapping:
% runtime rules
% variable re-labelling

We argue that the full potential of logic reasoning on the SW has not been realized.
Given that SPARQL is a well-established standard for SW querying, one way to meet this potential is to use SPARQL \textsc{construct} queries to express logic rules.
We implemented reasoning over \textsc{construct} queries by translating them into N3. 
% As N3 is a superset of Turtle, it is in line with SW principles; 
%Both SPARQL and 
As N3 natively operates on RDF triples and follows SW principles, translations can be exchanged between parties. % by, e.g., using IRIs for resource identification. 
Our translation maps 1 query to 1 rule, enabled by the novel concept of ``runtime'' rules, making the translations easier to exchange and manually tweak. 
% (wv) one of the reviewers used this to argue against our work, I think
% Together with the similarity between SPARQL and N3 syntaxes, this aspect can also simplify adoption by SW practitioners.
The runtime rules further allow using SPARQL features, namely \textsc{union} and \textsc{optional}, directly within ``vanilla'' N3 rules.
In fact, in the long term, we foresee the line between N3 and SPARQL syntax becoming blurred, and this translation is a first step in that direction.
% also means that users could manually tweak the translated rules.
% This paper presents a novel approach to bridge the gap between SPARQL and N3 for rule-based reasoning on the Semantic Web. By translating SPARQL queries, particularly \textsc{construct} queries, into N3 rules, our method enables recursive reasoning in a manner that aligns with Semantic Web principles. The proposed translation framework ensures that the expressivity of SPARQL is preserved while allowing compatibility with N3's rule-based reasoning capabilities. This work represents a significant step toward unifying querying and reasoning paradigms within the Semantic Web.
Our evaluation shows a competitive performance compared to state-of-the-art systems, and the benefits of having both forward and backward reasoning.
% highlights the benefits of general recursion in SPARQL (i.e., modularity); 
%and compares query vs. rule engines for complex data structures and reasoning over large ontologies.
% Our Zika use case further shows the added value of rule-based reasoning by adding general recursion to SPARQL, which enables modular and re-usable queries.

% Our evaluation shows that the proposed approach performs competitively with state-of-the-art systems, particularly for recursive reasoning tasks like transitive closures and complex rule applications. Its integration with biomedical ontologies demonstrates flexibility and scalability, with backward reasoning excelling in large ontologies and forward reasoning performing well with complex, structured data.

% The ability to maintain a one-to-one mapping between SPARQL queries and N3 rules ensures transparency and ease of adoption for Semantic Web practitioners. Users familiar with SPARQL can easily interpret and modify the translated N3 rules, facilitating wider acceptance and integration into existing workflows. Furthermore, the generated N3 rules remain exchangeable and interoperable within the broader Semantic Web ecosystem by adhering to Semantic Web principles.

\textit{Future Work}. 
% (wv) concretizing a bit
We aim to implement the missing SPARQL 1.1 features (see Section~\ref{sec:eval_setup}). %for full compliance with SPARQL 1.1.
We will formalize our translation of property paths and are currently investigating their efficient evaluation, for instance, by generating partially grounded, forward-chaining rules.
% We observed in our evaluation that \textit{spinrdf} performs much better for certain scenarios. 
The experiments showed that our rule translations are at a disadvantage for complex data structures (i.e., requiring many joins). %A major avenue of 
Future work involves the optimization of rule-based reasoning in such a setting; e.g., rule engine optimizations, and combining forward with backward reasoning.
Finally, we plan to extend the cross-fertilization between N3 and SPARQL, by supporting SPARQL expressions (e.g., \textsc{filter}) in N3 rules.

% Building on the current contributions, several directions for future research can be pursued. First, integrating support for SPARQL \textsc{filter} functions directly into native N3 rules would further enhance the expressivity and usability of the framework. 
% Additionally, extending the approach to support other SPARQL query forms, such as \textsc{ask} and \textsc{describe}, would broaden its applicability. 
% Investigating forward and backward reasoning optimizations, particularly for large and deeply nested datasets, is another promising avenue.

% \textit{Supplemental Material Statement}: The source code of our SiN3 implementation and all code to run our tests is available at github:  \url{https://anonymous.4open.science/r/research-paper-2025-F0CD/}. The datasets used can be downloaded from \url{https://gofile.io/d/Lb1J00}.
%\textit{Supplemental Material Statement}: The source code of SiN3 and our experiment setup is available at github: \url{https://anonymous.4open.science/r/research-paper-2025-F0CD/}. 
%To keep the submission anonymous, the datasets can currently be downloaded from \url{https://gofile.io/d/Lb1J00} (alternates: \url{https://files.catbox.moe/vpg5uy.zip}, \url{https://files.catbox.moe/06n7mv.zip}, \url{https://files.catbox.moe/shrw87.zip}).
 \bibliographystyle{splncs04}
\bibliography{bibliography}

\begin{thebibliography}{10}
\providecommand{\url}[1]{\texttt{#1}}
\providecommand{\urlprefix}{URL }
\providecommand{\doi}[1]{https://doi.org/#1}

\bibitem{sparqlex}
{SPARQL Exists} report. {W3C} community group report, W3C (Apr 2019), \url{https://w3c.github.io/sparql-exists/docs/sparql-exists.html}

\bibitem{spinrdf}
Andy~Seaborne, M.J.: Spinrdf (2019), \url{https://github.com/spinrdf/spinrdf}

\bibitem{angles2023}
Angles, R., Gottlob, G., Pavlovic, A., Pichler, R., Sallinger, E.: Sparqlog: {A} system for efficient evaluation of {SPARQL} 1.1 queries via datalog. Proc. {VLDB} Endow.  \textbf{16}(13),  4240--4253 (2023), \url{https://www.vldb.org/pvldb/vol16/p4240-sallinger.pdf}

\bibitem{Apache2021}
{Apache}: Apache {Jena} (2021), \url{https://jena.apache.org/}

\bibitem{arenas2018expressive}
Arenas, M., Gottlob, G., Pieris, A.: Expressive languages for querying the semantic web. ACM Trans. Database Syst.  \textbf{43}(3) (Nov 2018). \doi{10.1145/3238304}, \url{https://doi.org/10.1145/3238304}

\bibitem{sin3_demo}
Arndt, D., Woensel, W.V., Tomaszuk, D.: {SiN3}: Scalable inferencing with {SPARQL} {CONSTRUCT} queries. In: Fundulaki, I., Kozaki, K., Garijo, D., G{\'{o}}mez{-}P{\'{e}}rez, J.M. (eds.) Proceedings of the {ISWC} 2023 Posters, Demos and Industry Tracks: From Novel Ideas to Industrial Practice co-located with 22nd International Semantic Web Conference {(ISWC} 2023), Athens, Greece, November 6-10, 2023. {CEUR} Workshop Proceedings, vol.~3632. CEUR-WS.org (2023), \url{https://ceur-ws.org/Vol-3632/ISWC2023\_paper\_469.pdf}

\bibitem{arndt2023n3semantics}
Arndt, D., Champin, P.A.: Notation3 semantics. {W3C} community group report, W3C (Jul 2023), \url{https://w3c.github.io/N3/reports/20230703/semantics.html}

\bibitem{onlineappendix}
Arndt, D., Van~Woensel, W., Tomaszuk, D.: {SPARQL} in {N3} github repository, \url{https://github.com/domel/research-paper-2025}

\bibitem{atzori2014}
Atzori, M.: Computing recursive {SPARQL} queries. In: 2014 IEEE International Conference on Semantic Computing. pp. 258--259. IEEE (2014)

\bibitem{bellomarini2018vadalog}
Bellomarini, L., Sallinger, E., Gottlob, G.: The vadalog system: datalog-based reasoning for knowledge graphs. Proc. VLDB Endow.  \textbf{11}(9),  975–987 (May 2018). \doi{10.14778/3213880.3213888}, \url{https://doi.org/10.14778/3213880.3213888}

\bibitem{cwm}
Berners-Lee, T.: {Cwm} (2000--2009), \url{http://www.w3.org/2000/10/swap/doc/cwm.html}

\bibitem{N3Logic}
Berners-Lee, T., Connolly, D., Kagal, L., Scharf, Y., Hendler, J.: {N3Logic: A logical framework for the World Wide Web}. Theory and Practice of Logic Programming  \textbf{8}(3),  249--269 (2008). \doi{10.1017/S1471068407003213}, \url{http://arxiv.org/abs/0711.1533}

\bibitem{OWL2RL}
Calvanese, D., Carroll, J., De~Giacomo, G., Hendler, J., Herman, I., Parsia, B., Patel-Schneider, P.F., Ruttenberg, A., Sattler, U., Schneider, M.: {OWL2} web ontology language profiles (second edition): {OWL2} {RL}. W3c recommendation, W3C (Dec 2012), https://www.w3.org/TR/2012/REC-owl2-profiles-20121211/

\bibitem{corby2017ldscript}
Corby, O., Faron-Zucker, C., Gandon, F.: {LDScript}: a linked data script language. In: International Semantic Web Conference. pp. 208--224. Springer (2017)

\bibitem{fischer2010towards}
Fischer, F., Unel, G., Bishop, B., Fensel, D.: Towards a scalable, pragmatic knowledge representation language for the web. In: Perspectives of Systems Informatics: 7th International Andrei Ershov Memorial Conference, PSI 2009, Novosibirsk, Russia, June 15-19, 2009. Revised Papers 7. pp. 124--134. Springer (2010)

\bibitem{fokoue2012owl}
Fokoue, A., Horrocks, I., Motik, B., Grau, B.C., Wu, Z.: {OWL} 2 web ontology language profiles (second edition). {W3C} recommendation, W3C (Dec 2012), https://www.w3.org/TR/2012/REC-owl2-profiles-20121211/

\bibitem{gottlob2015beyond}
Gottlob, G., Pieris, A.: Beyond {SPARQL} under {OWL 2 QL} entailment regime: Rules to the rescue. In: Twenty-Fourth International Joint Conference on Artificial Intelligence (2015)

\bibitem{harris2013sparql}
Harris, S., Seaborne, A.: {SPARQL} 1.1 query language. {W3C} recommendation, W3C (Mar 2013), \url{https://www.w3.org/TR/2013/REC-sparql11-query-20130321/}

\bibitem{hassanzadeh2009linked}
Hassanzadeh, O., Consens, M.P.: Linked movie data base. In: LDOW (2009)

\bibitem{RDFsemantics}
Hayes, P., Patel-Schneider, P.F.: {RDF} 1.1 {Semantics}. W3c recommendation, W3C (Feb 2014), https://www.w3.org/TR/rdf11-mt/

\bibitem{HL7FHIR}
{HL7 International}: {HL7} {Fast} {Health} {Interop} {Resources} ({FHIR}), \url{https://www.hl7.org/index.cfm}

\bibitem{hogan2020recursive}
Hogan, A., Reutter, J., Soto, A.: Recursive {SPARQL} for graph analytics. arXiv preprint arXiv:2004.01816  (2020)

\bibitem{horrocks2004}
Horrocks, I., Patel-Schneider, P.F., Boley, H., Tabet, S., Grosof, B., Dean, M.: {SWRL}: A semantic web rule language combining owl and ruleml. Tech. rep., W3C (2004), \url{https://www.w3.org/submissions/SWRL/}

\bibitem{knublauchSPIN}
Knublauch, H., Hendler, J.A., Idehen, K.: {SPIN} - {Overview} and {Motivation}, \url{https://www.w3.org/submissions/spin-overview/}

\bibitem{motik2012}
Motik, B., Patel-Schneider, P., Grau, B.C.: {OWL} 2 web ontology language direct semantics (second edition). {W3C} recommendation, W3C (Dec 2012), https://www.w3.org/TR/2012/REC-owl2-direct-semantics-20121211/

\bibitem{pellissier2020yago}
Pellissier~Tanon, T., Weikum, G., Suchanek, F.: {Yago 4}: A reason-able knowledge base. In: The Semantic Web: 17th International Conference, ESWC 2020, Heraklion, Crete, Greece, May 31--June 4, 2020, Proceedings 17. pp. 583--596. Springer (2020)

\bibitem{perez2009semantics}
P{\'e}rez, J., Arenas, M., Gutierrez, C.: Semantics and complexity of sparql. ACM Transactions on Database Systems ({TODS})  \textbf{34}(3),  1--45 (2009)

\bibitem{polleres2007}
Polleres, A.: From {SPARQL} to rules (and back). In: Proceedings of the 16th international conference on World Wide Web. pp. 787--796 (2007)

\bibitem{polleres2013relation}
Polleres, A., Wallner, J.P.: On the relation between {SPARQL} 1.1 and answer set programming. Journal of Applied Non-Classical Logics  \textbf{23}(1-2),  159--212 (2013)

\bibitem{reutter2021}
Reutter, J., Soto, A., Vrgo{\v{c}}, D.: Recursion in {SPARQL}. Semantic Web  \textbf{12}(5),  711--740 (2021)

\bibitem{salas2022semantics}
Salas, J., Hogan, A.: Semantics and canonicalisation of {SPARQL} 1.1. Semantic Web  \textbf{13}(5),  829--893 (2022)

\bibitem{sparqlspec}
{SPARQL Working Group}: {SPARQL} 1.1 {Query} {Language} (2013), \url{https://www.w3.org/TR/sparql11-query}

\bibitem{tomaszuk2018inference}
Tomaszuk, D.: Inference rules for {OWL-P} in n3logic. In: FedCSIS (Communication Papers). pp. 27--33 (2018)

\bibitem{SNOMED}
{U.S. National Library of Medicine}: {SNOMED} {CT}, \url{https://www.nlm.nih.gov/healthit/snomedct/index.html}

\bibitem{jen3}
Van~Woensel, W.: Jen3. https://github.com/william-vw/jen3

\bibitem{woensel2023n3}
{Van Woensel}, W., Arndt, D., Champin, P.A., Tomaszuk, D., Kellogg, G.: Notation3 language. {W3C} community group report, W3C (Jul 2023), \url{https://w3c.github.io/N3/reports/20230703/}

\bibitem{eyepaper}
Verborgh, R., De~Roo, J.: Drawing conclusions from linked data on the web: The {EYE} reasoner. IEEE Software  \textbf{32}(5),  23--27 (May 2015)

\end{thebibliography}

 %\appendix
% \input{appendices}
%\input{correctness-proof}
\end{document}